\documentclass[runningheads]{llncs}

\usepackage{xspace,amsmath,amssymb,bbm}
\usepackage[sort]{cite}
\usepackage{gastex}
\usepackage{tikz}
\usetikzlibrary{arrows,automata,shapes,calc}
\usepackage[matrix,arrow,cmtip]{xy}
\SelectTips{cm}{}\SilentMatrices\CompileMatrices
\renewcommand*\to[1][]{\ensuremath{\xrightarrow{#1}}}
\newcommand*\Int{\mathbbm{Z}}
\newcommand*\Nat{\mathbbm{N}}
\newcommand*\Rat{\mathbbm{Q}}
\newcommand*\Natp{\Nat_+}
\newcommand*\Real{\mathbbm{R}}
\newcommand*\Realnn{\Real_{ \ge 0}}
\newcommand*\id{\textup{\textsf{id}}}
\newcommand*\bigmid{\mathrel{\big|}}
\newcommand*\ie{\textit{i.e.}\xspace}
\newcommand*\cf{\textit{cf.}\xspace}
\newcommand*\eg{\textit{e.g.}\xspace}
\newcommand*\IMPL{\Rightarrow}
\newcommand*\TRUE{\textup{\textbf{t\!t}}}
\newcommand*\FALSE{\textup{\textbf{ff}}}
\newcommand*\rest[1]{_{\upharpoonleft #1}}
\newcommand{\state}{\textup{\textsf{state}}}
\newcommand{\val}{\textup{\textsf{val}}}
\newcommand*\E{\mathcal E}
\newcommand*\Eint{\E_\textup{\textsf{int}}}
\newcommand*\Eflat{\bar\E}
\newcommand*\Epw{\E_{\textup{\textsf{pw}}}}
\newcommand*\Epwi{\E_{\textup{\textsf{pwi}}}}
\newcommand*\Eflatpw{\Eflat_{\textup{\textsf{pw}}}}
\newcommand*\Eflatpwi{\Eflat_{\textup{\textsf{pwi}}}}
\newcommand*\pbuchi{\textup{\textsf{B\"uchi}}}
\newcommand*\preach{\textup{\textsf{Reach}}}
\newcommand*\transpose[1]{{\vphantom{#1}}_{\textup{\textsf{t}}}{#1}}
\newcommand*\bbot{%
  \hbox{%
    \hbox to -.475pt{%
      \raisebox{1.5pt}[0pt][0pt]{%
        \rule{6.4pt}{.4pt}%
      }%
      \hss}%
    $\bot$%
  }}
\newcommand*\ttop{%
  \hbox{%
    \hbox to -.475pt{%
      \raisebox{5pt}[0pt][0pt]{%
        \rule{6.4pt}{.4pt}%
      }%
      \hss}%
    $\top$%
  }}
\newcommand*\V{\mathcal V}
\newcommand*\B{\mathcal B}
\begin{document}

\title{Kleene Algebras and Semimodules \\ for Energy Problems}
\titlerunning{Kleene Algebras and Semimodules for Energy Problems}

\author{%
  Zolt{\'a}n {\'E}sik\inst1\thanks{The reserach of this author is
    supported by the European Union and co-funded by the European Social
    Fund. Project title: `Telemedicine-focused research activities on
    the field of mathematics, informatics and medical sciences'.
    Project number: T\'AMOP-4.2.2.A-11/1/KONV-2012-0073} \and Uli
  Fahrenberg \inst2 \and Axel Legay \inst2 \and Karin Quaas \inst3 }
\authorrunning{{\'E}sik, Fahrenberg, Legay, Quaas} \institute{
  University of Szeged, Hungary \and Irisa / INRIA Rennes, France \and
  Universit\"at Leipzig, Germany }
\maketitle

\begin{abstract}
  With the purpose of unifying a number of approaches to energy problems
  found in the literature, we introduce generalized energy automata.
  These are finite automata whose edges are labeled with energy
  functions that define how energy levels evolve during transitions.
  Uncovering a close connection between energy problems and reachability
  and B{\"u}chi acceptance for semiring-weighted automata, we show that
  these generalized energy problems are decidable.  We also provide
  complexity results for important special cases.
\end{abstract}

\section{Introduction}

\emph{Energy} and \emph{resource management} problems are important in
areas such as embedded systems or autonomous systems.  They are
concerned with the question whether a given system admits infinite
schedules during which (1) certain tasks can be repeatedly accomplished
and (2) the system never runs out of energy (or other specified
resources).  Starting with~\cite{DBLP:conf/formats/BouyerFLMS08}, formal
modeling and analysis of such problems has recently attracted some
attention~\cite{conf/ictac/FahrenbergJLS11, DBLP:conf/lata/Quaas11,
  DBLP:conf/icalp/ChatterjeeD10, DBLP:conf/hybrid/BouyerFLM10,
  DBLP:conf/qest/BouyerLM12, DBLP:conf/csl/DegorreDGRT10}.

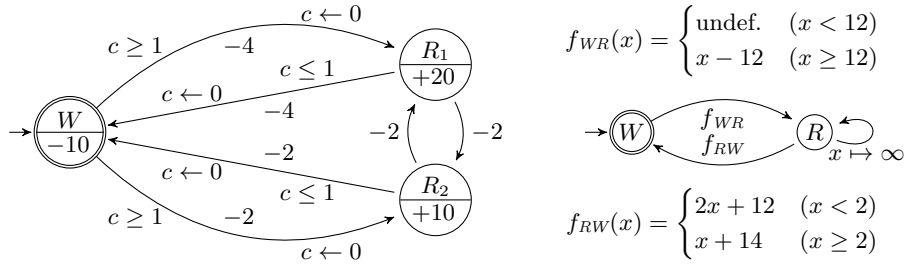
\begin{figure}[tbp]
  \centering
  \begin{tikzpicture}[->,>=stealth',shorten >=1pt,auto,node
    distance=2.0cm,initial text=,scale=.9]
    \tikzstyle{every node}=[font=\small]
    \tikzstyle{every state}=[inner sep=.5mm,minimum size=2.2mm,outer
    sep=.8mm]
    \begin{scope}[xscale=.9]
      \node[state with output,initial,accepting] (W) at (0,0) {$W$
        \nodepart{lower} $\!-10$};
      \node[state with output] (R1) at (6,1) {$R_1$ \nodepart{lower}
        $\!+20$};
      \node[state with output] (R2) at (6,-1) {$R_2$ \nodepart{lower}
        $\!+10$};
      \path (W) edge [bend left] node [below] {$-4$} node
      [above,pos=.2] {$c\ge 1\,\,\,\,\,\,\,\,$} node [above,pos=.8]
      {$c\leftarrow 0$} (R1);
      \path (R1) edge node [below,pos=.4] {$-4$} node
      [above,pos=.3] {$c\le 1$} node [above,pos=.7] {$c\leftarrow 0$}
      (W);
      \path (W) edge [bend right] node [above] {$-2$} node
      [below,pos=.2] {$c\ge 1\,\,\,\,\,\,\,\,$} node [below,pos=.8]
      {$c\leftarrow 0$} (R2);
      \path (R2) edge node [above,pos=.4] {$-2$} node
      [below,pos=.3] {$c\le 1$} node [below,pos=.7] {$c\leftarrow 0$}
      (W);
      \path (R2) edge [bend left] node [left] {$-2$} (R1);
      \path (R1) edge [bend left] node [right] {$-2$} (R2);
    \end{scope}
    \begin{scope}[xshift=8.3cm,scale=.9]
      \node[state,initial,accepting] (W) at (0,0) {$W$};
      \node[state] (R) at (3,0) {$R$};
      \path (W) edge [bend left] node [below] {$f_\textit{WR}$} (R);
      \path (R) edge [bend left] node [above] {$f_\textit{RW}$} (W);
      \path (R) edge [out=-20,in=20,loop] node [below,pos=.3] {$x\mapsto
        \infty$} (R);
      \node at (1.5,1.5) {%
        $f_\textit{WR}( x) =
        \begin{cases}
          \text{undef.} & ( x< 12) \\
          x- 12 & ( x\ge 12)
        \end{cases}$};
      \node at (1.5,-1.5) {%
        $f_\textit{RW}( x) =
        \begin{cases}
          2 x+ 12 & ( x< 2) \\
          x+ 14 & ( x\ge 2)
        \end{cases}$};
    \end{scope}
  \end{tikzpicture}
  \caption{%
    \label{fi:example}
    Simple model of an electric car as a weighted timed automaton (left);
    the corresponding energy automaton (right)}
\end{figure}

As an example, the left part of Fig.~\ref{fi:example} shows a simple
model of an electric car, modeled as a weighted timed
automaton~\cite{DBLP:conf/hybrid/AlurTP01,
  DBLP:conf/hybrid/BehrmannFHLPRV01}.  In the \emph{working} state $W$,
energy is consumed at a rate of $10$ energy units per time unit; in the
two \emph{recharging} states $R_1$, $R_2$, the battery is charged at a
rate of $20$, respectively $10$, energy units per time unit.  As the
clock $c$ is reset ($c\leftarrow 0$) when entering state $W$ and has
guard $c\ge 1$ on outgoing transitions, we ensure that the car always
has to be in state $W$ for at least one time unit.  Similarly, the
system can only transition back from states $R_1$, $R_2$ to $W$ if it
has spent at most one time unit in these states.

Passing between states $W$ and $R_1$ requires $4$ energy units, while
transitioning between $W$ and $R_2$, and between $R_2$ to $R_1$,
requires $2$ energy units.  Altogether, this is intended to model the
fact that there are two recharge stations available, one close to work
but less powerful, and a more powerful one further away.  Now assume
that the initial state $W$ is entered with a given \emph{initial energy}
$x_0$, then the energy problem of this model is as follows: Does there
exist an infinite trace which (1) visits $W$ infinitely often and (2)
never has an energy level below $0$?

This type of energy problems for weighted timed automata is treated
in~\cite{DBLP:conf/hybrid/BouyerFLM10}, and using a reduction like
in~\cite{DBLP:conf/hybrid/BouyerFLM10}, our model can be transformed to
the \emph{energy automaton} in the right part of Fig.~\ref{fi:example}.
(The reduction is quite complicated and only works for one-clock timed
automata; see~\cite{DBLP:conf/hybrid/BouyerFLM10} for details.)  It can
be shown that the energy problem for the original automaton is
equivalent to the following problem in the energy automaton: Given an
initial energy $x_0$, and updating the energy according to the
transition label whenever taking a transition, does there exist an
infinite run which visits $W$ infinitely often?  Remark that the energy
update on the transition from $R$ to $W$ is rather complex (in the
general case of $n$ recharge stations, the definition of $f_\textit{RW}$
can have up to $n$ branches), and that we need to impose a B{\"u}chi
condition to enforce visiting $W$ infinitely often.

In this paper we propose a generalization of the energy automata
of~\cite{DBLP:conf/hybrid/BouyerFLM10} which also encompasses most other
approaches to energy problems.  Abstracting the properties of the
transition update functions in our example, we define a general notion
of \emph{energy functions} which specify how weights change from one
system state to another.  Noticing that our functional energy automata
are semiring-weighted automata in the sense of~\cite{book/DrosteKV09},
we uncover a close connection between energy problems and reachability
and B{\"u}chi problems for weighted automata.  More precisely, we show
that one-dimensional energy problems can be naturally solved using
matrix operations in semirings and
semimodules~\cite{book/DrosteKV09,BEbook,EK2,DBLP:conf/mfcs/EsikK04}.

For reachability, we use only standard results~\cite{book/DrosteKV09},
but for B{\"u}chi acceptance we have to extend previous
work~\cite{EK2,DBLP:conf/mfcs/EsikK04} as our semiring is not complete.
We thus show that reachability and B{\"u}chi acceptance are decidable
for energy automata.  For the class of \emph{piecewise affine} energy
functions, which generalize the functions of Fig.~\ref{fi:example} and
are important in applications, they are decidable in exponential time.
% , and we also provide a PSPACE lower bound.

\smallskip\noindent\textit{Structure of the Paper.}\quad %
We introduce our general model of energy automata in
Section~\ref{se:energyaut}.  In Section~\ref{se:esemiring} we show that
the set of energy functions forms a star-continuous Kleene algebra, a
fact which allows us to give an elegant characterization of reachability
in energy automata.  We also expose a structure of Conway
semiring-semimodule pair over energy functions which permits to
characterize B{\"u}chi acceptance.  In Section~\ref{se:reach} we use
these characterizations to prove that reachability and B{\"u}chi
acceptance are decidable for energy automata.  We also show that this
result is applicable to most of the above-mentioned examples and give
complexity bounds.  To put our results in perspective, we generalize
energy automata along several axes in Section~\ref{se:multi} and analyze
these generalized reachability and B{\"u}chi acceptance problems.
% Owing to space limitations, most of the proofs had to be omitted from
% this paper; these can be found in the extended
% version~\cite{arxiv/EsikFLQ13}.

\smallskip\noindent\textit{Related Work.}\quad %
A simple class of energy automata is the one of \emph{integer-weighted
  automata}, where all energy functions are updates of the form
$x\mapsto x+ k$ for some (positive or negative) integer $k$.  Energy
problems on these automata, and their extensions to multiple weights
(also called \emph{vector addition systems with states} (VASS)) and
games, have been considered \eg~in~\cite{conf/ictac/FahrenbergJLS11,
  DBLP:conf/fsttcs/ChatterjeeDHR10, DBLP:conf/icalp/ChatterjeeD10,
  DBLP:conf/icalp/BrazdilJK10, DBLP:conf/rp/Chaloupka10,
  DBLP:journals/ipl/Chan88, DBLP:conf/formats/BouyerFLMS08}.  Our energy
automata may hence be considered as a generalization of one-dimensional
VASS to arbitrary updates; in the final section of this paper we will
also be concerned with multi-dimensional energy automata and games.

Energy problems on \emph{timed
  automata}~\cite{DBLP:journals/tcs/AlurD94} have been considered
in~\cite{DBLP:conf/lata/Quaas11, DBLP:conf/formats/BouyerFLMS08,
  DBLP:conf/hybrid/BouyerFLM10, DBLP:conf/qest/BouyerLM12}.  Here timed
automata are enriched with integer weights in locations and on
transitions (the \emph{weighted timed automata}
of~\cite{DBLP:conf/hybrid/AlurTP01, DBLP:conf/hybrid/BehrmannFHLPRV01},
\cf~Fig.~\ref{fi:example}), with the semantics that the weight of a
delay in a location is computed by multiplying the length of the delay
by the location weight.  In~\cite{DBLP:conf/formats/BouyerFLMS08} it is
shown that energy problems for one-clock weighted timed automata without
updates on transitions (hence only with weights in locations) can be
reduced to energy problems on integer-weighted automata with additive
updates.

For one-clock weighted timed automata \emph{with} transition updates,
energy problems are shown decidable
in~\cite{DBLP:conf/hybrid/BouyerFLM10}, using a reduction to energy
automata as we use them here.  More precisely, each path in the timed
automaton in which the clock is not reset is converted to an edge in an
energy automaton, labeled with a \emph{piecewise affine} energy function
(\cf~Definition~\ref{de:fpwint}).  Decidability of the energy problem is
then shown using ad-hoc arguments, but can easily be inferred from our
general results in the present paper.

Also another class of energy problems on weighted timed automata is
considered in~\cite{DBLP:conf/hybrid/BouyerFLM10}, in which weights
during delays are increasing \emph{exponentially} rather than linearly.
These are shown decidable using a reduction to energy automata with
\emph{piecewise polynomial} energy functions; again our present
framework applies.

We also remark that semigroups acting on a set, or more generally,
semiring-semimodule pairs, have been used to describe the infinitary
behavior of automata for a long time, see~\cite{book/PerrinP04,
  DBLP:conf/icalp/Wilke91, BEbook}. In this framework, the infinitary
product or omega operation is defined on the semiring and takes its
values in the semimodule. Another approach is studied
\eg~in~\cite{DBLP:conf/RelMiCS/MathieuD05}, where the omega operation
maps the semiring into itself.  It seems to the authors that there is no
reasonable definition of an infinitary product or omega operation on
energy functions that would again result in an energy function, hence we
chose to use the framework of semiring-semimodule pairs.

\section{Energy Automata}
\label{se:energyaut}

The transition labels on the energy automata which we consider in the
paper, will be functions which model transformations of energy levels
between system states.  Such transformations have the (natural)
properties that below a certain energy level, the transition might be
disabled (not enough energy is available to perform the transition), and
an increase in input energy always yields at least the same increase in
output energy.  Thus the following definition.

\begin{definition}
  An \emph{energy function} is a partial function $f: \Realnn\to
  \Realnn$ which is defined on a closed interval $[ l_f,
  \infty\mathclose[$ or on an open interval $\mathopen] l_f,
  \infty\mathclose[$, for some lower bound $l_f\ge 0$, and such that for
  all $x_1\le x_2$ for which $f$ is defined,
  \begin{equation}
    \label{eq:deriv1}\tag{$\ast$}
    f( x_2)\ge f( x_1)+ x_2- x_1\,.
  \end{equation}
  The class of all energy functions is denoted by $\mathcal F$.
\end{definition}

Thus energy functions are strictly increasing, and in points where they
are differentiable, the derivative is at least $1$.\footnote{Remark
  that, in relation to the example in the introduction, the derivative
  is taken with respect to \emph{energy input}, not \emph{time}.  Hence
  the mapping from input to output energy in state $W$ is indeed an
  energy function in our sense.}  The inverse functions to energy
functions exist, but are generally not energy functions.  Energy
functions can be \emph{composed}, where it is understood that for a
composition $g\circ f$ (to be read from right to left), the interval of
definition is $\{ x\in \Realnn\mid f( x)\text{ and } g( f( x))\text{
  defined}\}$.  We will generally omit the symbol $\circ$ and write
composition simply as $gf$.

\begin{definition}
  An \emph{energy automaton} $( S, T)$ consists of finite sets $S$ of
  states and $T\subseteq S\times \mathcal F\times S$ of transitions
  labeled with energy functions.
\end{definition}

\begin{figure}[tbp]
  \centering
  \begin{tikzpicture}[->,>=stealth',shorten >=1pt,auto,node
    distance=2.0cm,initial text=,scale=.8,transform shape]
    \tikzstyle{every node}=[font=\small]
    \tikzstyle{every state}=[fill=white,shape=circle,inner
    sep=.5mm,minimum size=2.2mm,outer sep=.8mm]
    \node[state,initial] (1) at (0,0) {};
    \node[state] (2) at (4,0) {};
    \node[state] (3) at (8,0) {};
    \path (1) edge[out=10,in=170] node[above] {$x\mapsto x+ 2; x\ge 2$}
    (2);
    \path (1) edge[out=-10,in=-170] node[below] {$x\mapsto x+ 3; x> 1$}
    (2);
    \path (2) edge[out=120,in=60,loop] node[above] {$x\mapsto 2x- 2; x\ge
      1$} (2);
    \path (2) edge[out=10,in=170] node[above] {$x\mapsto x- 1; x> 1$}
    (3);
    \path (3) edge[out=190,in=-10] node[below] {$x\mapsto x+ 1;
      x\ge 0$}
    (2);
  \end{tikzpicture}
  \caption{%
    \label{fi:eauto}
    A simple energy automaton.}
\end{figure}
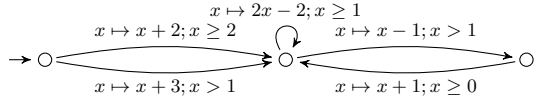

We show an example of a simple energy automaton in Fig.~\ref{fi:eauto}.
Here we use inequalities to give the definition intervals of energy
functions.

A finite \emph{path} in an energy automaton is a finite sequence of
transitions $\pi= (s_0,f_1,s_1), (s_1,f_2,s_2),\dots,
(s_{n-1},f_n,s_n)$.  We use $f_\pi$ to denote the combined energy
function $f_n\cdots f_2 f_1$ of such a finite path.  We will also use
infinite paths, but note that these generally do not allow for combined
energy functions.

A \emph{global state} of an energy automaton is a pair $q=( s, x)$ with
$s\in S$ and $x\in \Realnn$.  A transition between global states is of
the form $((s,x), f, (s',x'))$ such that $(s,f,s')\in T$ and $x'=f(x)$.
A (finite or infinite) \emph{run} of $( S, T)$ is a path in the graph of
global states and transitions.

We are ready to state the decision problems with which our main concern
will lie.  As the input to a decision problem must be in some way
finitely representable, we will state them for subclasses $\mathcal
F'\subseteq \mathcal F$ of \emph{computable} energy functions; an
$\mathcal F'$-automaton is an energy automaton $( S, T)$ with
$T\subseteq S\times \mathcal F'\times S$.

\begin{problem}[Reachability]
  \label{pb:reach}
  Given a subset $\mathcal F'$ of computable functions, an $\mathcal
  F'$-automaton $( S, T)$, an initial state $s_0\in S$, a set of
  accepting states $F\subseteq S$, and a computable initial energy $x_0\in
  \Realnn$: does there exist a finite run of $( S, T)$ from $(s_0,x_0)$
  which ends in a state in $F$?
\end{problem}

\begin{problem}[B{\"u}chi acceptance]
  \label{pb:buchi}
  Given a subset $\mathcal F'$ of computable functions, an $\mathcal
  F'$-automaton $( S, T)$, an initial state $s_0\in S$, a set of
  accepting states $F\subseteq S$, and a computable initial energy
  $x_0\in \Realnn$: does there exist an infinite run of $( S, T)$ from
  $(s_0,x_0)$ which visits $F$ infinitely often?
\end{problem}

As customary, a run such as in the statements above is said to be
accepting.  We let $\preach_{ \mathcal F'}$ denote the function which
maps an $\mathcal F'$-automaton together with an initial state, a set of
final states, and an initial energy to the Boolean values $\FALSE$ or
$\TRUE$ depending on whether the answer to the concrete reachability
problem is negative or positive.  $\pbuchi_{ \mathcal F'}$ denotes the
similar mapping for B{\"u}chi problems.

The special case of Problem~\ref{pb:buchi} with $F= S$ is the question
whether there \emph{exists an infinite run} in the given energy
automaton.  This is what is usually referred to as \emph{energy
  problems} in the literature; our extension to general B{\"u}chi
conditions has not been treated before.

\section{The Algebra of Energy Functions}
\label{se:esemiring}

In this section we develop an algebraic framework of
\emph{star-continuous Kleene algebra} around energy functions which will
allow us to solve reachability and B{\"u}chi acceptance problems in a
generic way.  Let $[ 0, \infty]_\bot=\{ \bot\}\cup[ 0, \infty]$ denote
the non-negative real numbers together with extra elements $\bot$,
$\infty$, with the standard order on $\Realnn$ extended by $\bot< x<
\infty$ for all $x\in \Realnn$.  Also, $\bot+ x= \bot- x= \bot$ for all
$x\in \Realnn\cup\{ \infty\}$ and $\infty+ x= \infty- x$ for all $x\in
\Realnn$.

\begin{definition}
  An \emph{extended energy function} is a mapping $f:[ 0,
  \infty]_\bot\to[][ 0, \infty]_\bot$, for which $f( \bot)= \bot$ and
  $f( x_2)\ge f( x_1)+ x_2- x_1$ for all $x_1\le x_2$, as
  in~\eqref{eq:deriv1}.  Moreover, $f( \infty)= \infty$, unless $f( x)=
  \bot$ for all $x\in[ 0, \infty]_\bot$.  The class of all extended
  energy functions is denoted $\E$.
\end{definition}

This means, in particular, that $f( x)= \bot$ implies $f( x')= \bot$ for
all $x'\le x$, and $f( x)= \infty$ implies $f( x')= \infty$ for all
$x'\ge x$.  Hence, except for the extension to $\infty$, these functions
are indeed the same as our energy functions from the previous section.
Composition of extended energy functions is defined as before, but needs
no more special consideration about its definition interval.

We also define an ordering on $\E$, by $f\le g$ iff $f( x)\le g( x)$ for
all $x\in[ 0, \infty]_\bot$.  We will need three special energy
functions, $\bbot$, $\id$ and $\ttop$; these are given by $\bbot( x)=
\bot$, $\id( x)= x$ for $x\in[ 0, \infty]_\bot$, and $\ttop( \bot)=
\bot$, $\ttop( x)= \infty$ for $x\in[ 0, \infty]$.

\begin{lemma}
  \label{le:distr}
  With the ordering $\le$, $\E$ is a complete lattice with bottom
  element $\bbot$ and top element $\ttop$.  The supremum on $\E$ is
  pointwise, \ie~$( \sup_{ i\in I} f_i)( x)= \sup_{ i\in I} f_i( x)$ for
  any set $I$, all $f_i\in \E$ and $x\in[ 0, \infty]_\bot$.  Also,
  $(\sup_{i \in I}f_i)h = \sup_{i \in I}( f_i h)$ for all $h\in \E$.
\end{lemma}

We denote binary suprema using the symbol $\vee$; hence $f\vee g$, for
$f, g\in \E$, is the function $( f\vee g)( x)= \max( f( x), g( x))$.

\begin{lemma}
  \label{le:idmsring}
  $( \E, \vee, \circ, \bbot, \id)$ is an idempotent semiring with natural
  order $\le$.
\end{lemma}

Recall~\cite{book/DrosteKV09} that $\le$ being natural refers to the
fact that $f\le g$ iff $f\vee g= g$.

For iterating energy functions, we define a unary star operation on $\E$
by
\begin{equation*}
  f^*( x)= \left\{
  \begin{array}{cl}
    x &\quad\text{if } f( x)\le x\,, \\
    \infty &\quad\text{if } f( x)> x\,.
  \end{array} \right.
\end{equation*}

\begin{lemma}
  \label{le:prop-star}
  For any $f\in \E$, we have $f^*\in \E$.  Also, for any
  $g\in \E$, there exists $f\in \E$ such that $g= f^*$
  if, and only if, there is $k\in[ 0, \infty]_\bot$ such that $g( x)= x$
  for all $x< k$, $g( x)= \infty$ for all $x> k$, and $g( k)= k$ or $g(
  k)= \infty$.
\end{lemma}

By Lemma~\ref{le:distr}, composition right-distributes over arbitrary
suprema in $\E$.  The following example shows that a similar left
distributivity does \emph{not} hold in general, hence $\E$ is \emph{not}
a complete semiring the sense of~\cite{book/DrosteKV09}.  Let $f_n,
g\in \E$ be defined by $f_n( x)= x+ 1- \frac1n$ for $x\ge 0$, $n\in
\Natp$ and $g( x)= x$ for $x\ge 1$.  Then $g( \sup_n f_n)( 0)= g( \sup_n
f_n( 0))= g( 1)= 1$, whereas $( \sup_n g f_n)( 0)= \sup_n g( f_n( 0))=
\sup_n g( 1- \frac1n)= \bot$.

The next lemma shows a restricted form of left distributivity which
holds only for function powers $f^n$.  Note that it implies that $f^*=
\sup_n f^n$ for all $f\in \E$, which justifies the definition of $f^*$
above.

\begin{lemma}
  \label{le:distl}
  For any $f, g\in \E$, $g f^*= \sup_{ n\in \Nat}( g f^n)$.
\end{lemma}

\begin{proposition}
  \label{pr:EcontKA}
  For any $f, g, h\in \E$, $g f^* h= \sup_{ n\in \Nat}( g f^n h)$.
  Hence $\E$ is a star-continuous Kleene
  algebra~\cite{DBLP:conf/mfcs/Kozen90}.
\end{proposition}

We call a subsemiring $\E'\subseteq \E$ a \emph{subalgebra} if $f^*\in
\E'$ for all $f\in \E'$.

It is known~\cite{book/DrosteKV09, book/Conway71, BEbook,
  DBLP:journals/iandc/Kozen94} that when $S$ is a star-continuous Kleene
algebra, then so is any matrix semiring $S^{n \times n}$, for all $n\ge
1$, with the usual sum and product operations. The natural order on
$S^{n \times n}$ is pointwise, so that for all $n\times n$ matrices
$A,B$ over $S$, $A\le B$ iff $A_{i,j}\le B_{i,j}$ for all $i,j$.  Now a
star-continuous Kleene algebra is also a Conway semiring, hence the
Conway identities
\begin{equation}
  \label{eq:conway1}
  (g \vee f)^*= (g^*f)^*g^* \quad\text{and}\quad (gf)^*= g(fg)^*f\vee \id
\end{equation}
are satisfied for all $f, g\in \E$.  Also, this implies that the matrix
semiring $\E^{ n\times n}$ is again a Conway semiring, for any $n\ge 1$,
with the star operation defined inductively for a matrix
\begin{equation}
  \label{eq:matrix}
  M=
  \begin{bmatrix}
    a & b\\
    c & d
  \end{bmatrix}
  \in \E^{n \times n}\,,
\end{equation}
where $a$ is $k \times k $ and $d$ is $m \times m$ with $k + m  = n$, by
\begin{equation}
  \label{eq:starM}
  M^*= 
  \begin{bmatrix}
    ( a\vee b d^* c)^* \;\;&\;\; ( a\vee b d^* c)^* b d^* \\
    ( d\vee c a^* b)^* c a^* \;\;&\;\; ( d\vee c a^* b)^*
  \end{bmatrix}
  \in \E^{n \times n}\,.
\end{equation}
The definition of $M^*$ does not depend on how $M$ is split into parts,
and star-continuity implies that for all matrices $M, N, O$,
\begin{equation}
  \label{eq:starcontM}
  N M^* O= \sup_{ n\in \Nat}( N M^n O)\,.
\end{equation}

Note again that this implies that $M^*= \sup_n M^n$ for all matrices
$M$.  In a sense, this gives another, inductive definition of the star
operation on the matrix semiring; the important property of
star-continuous Kleene algebras is, then, that this inductive definition
and the one in~\eqref{eq:starM} give rise to the same operation.

We introduce a semimodule $\V$ over $\E$.  Let $\B=\{ \FALSE, \TRUE\}$
be the Boolean algebra, with order $\FALSE< \TRUE$, and $\V=\{ u:[ 0,
\infty]_\bot\to \B\mid u( \bot)= \FALSE, x_1\le x_2\IMPL u( x_1)\le u(
x_2)\}$.  Identifying $\FALSE$ with $\bot$ and $\TRUE$ with $\infty$, we
have an embedding of $\V$ into $\E$; note that $\bbot, \ttop\in \V$.

\begin{lemma}
  \label{le:semimod}
  With action $( u, f)\mapsto u f: \V\times \E\to \V$, $\V$ is a right
  $\E$-semimodule~\cite{EK2}.  Moreover, $( \sup_{ i\in I} u_i) f=
  \sup_{ i\in I}( u_i f)$ for any set $I$, all $u_i\in \V$ and $f\in
  \E$, and $u f^*= \sup_{ n\in \Nat} u f^n$ for all $u\in \V$.
\end{lemma}

So like the situation for $\E$ (\cf~Lemmas~\ref{le:distr}
and~\ref{le:distl}), the action of $\E$ on $\V$ right-distributes over
arbitrary suprema and left-distributes over function powers.

We define an infinitary product operation $\E^\omega\to \V$.  Let $f_0,
f_1,\dotsc$ be an infinite sequence of energy functions and $x_0\in[ 0,
\infty]_\bot$, and put $x_{ n+ 1}= f_n( x_n)$ for $n\in \Nat$.  Then we
define
\begin{equation*}
  \big( \prod_{ i= 0}^\infty f_i)( x_0)=
  \begin{cases}
    \FALSE &\text{if } \exists n\in \Nat: x_n= \bot\,, \\
    \TRUE &\text{if } \forall n\in \Nat: x_n\ne \bot\,.
  \end{cases}
\end{equation*}
Note that this product is order-preserving.  By the next lemma, it is a
conservative extension of the finite product.  As $\E$ is not a complete
semiring, it follows that $( \E, \V)$ is not a complete
semiring-semimodule pair in the sense of~\cite{EK2}.

\begin{lemma}
  \label{le:prod_is_prod}
  For all $f_0, f_1,\dotsc\in \E$, $( \prod_{ i= 1}^\infty f_i) f_0=
  \prod_{ i= 0}^\infty f_i$.  For all indices $0= n_0\le n_1\le\dotsc$,
  $\prod_{ i= 0}^\infty f_i= \prod_{ i= 0}^\infty( f_{ n_{ i+ 1}-
    1}\dotsm f_{ n_i})$.
\end{lemma}

To deal with infinite iterations of energy functions, we define a unary
omega operation $\E\to \V$ by
\begin{equation*}
  f^\omega( x)= \left\{
  \begin{array}{cl}
    \FALSE &\quad\text{if } x= \bot\text{ or } f( x)< x\,, \\
    \TRUE &\quad\text{if } x\ne \bot\text{ and } f( x)\ge x\,.
  \end{array} \right.
\end{equation*}

Note that $f^\omega= \prod_{ i= 0}^\infty f$ for all $f\in \E$.

\begin{proposition}
  \label{pr:smodpair}
  $( \E, \V)$ is a Conway semiring-semimodule pair.
\end{proposition}

Recall~\cite{BEbook} that this means that additionally to the
identities~\eqref{eq:conway1},
\begin{equation*}
  ( g f)^\omega=( fg)^\omega f \quad\text{and}\quad ( f\vee g)^\omega=
  f^\omega( g f^*)^*\vee( g f^*)^\omega
\end{equation*}
for all $f, g\in \E$.  Like for Conway semirings, it implies that the
pair $( \E^{ n\times n}, \V^n)$ is again a Conway semiring-semimodule
pair, for any $n\ge 1$, with the action of $\E^{ n\times n}$ on $\V^n$
similar to matrix-vector multiplication using the action of $\E$ on
$\V$, and the omega operation $\E^{ n\times n}\to \V^n$ given
inductively as follows: for $M\in \E^{ n\times n}$ with blocks as
in~\eqref{eq:matrix}, define
\begin{gather}
  \label{eq:omegaM}
  M^\omega =
  \begin{bmatrix}
    ( a\vee b d^* c)^\omega\vee d^\omega c( a\vee b d^* c)^* \;\;&\;\; (
    d\vee c a^* b)^\omega\vee a^\omega b( d\vee c a^* b)^*
  \end{bmatrix}, \\
  M^{ \omega_k} =
  \begin{bmatrix}
    ( a\vee b d^* c)^\omega \;\;&\;\; ( a\vee b d^* c)^\omega b d^*
  \end{bmatrix}. \notag
\end{gather}
The definition of $M^\omega$ does not depend on how $M$ is split into
parts, but the one of $M^{ \omega_k}$ does (recall that $a$ is a
$k\times k$ matrix).  It can be shown~\cite{BEbook}
that~\eqref{eq:omegaM}, and also~\eqref{eq:starM}, follow directly from
certain general properties of fixed point operations.

\section{Decidability}
\label{se:reach}

We are now ready to apply the Kleene algebra framework to reachability
and B{\"u}chi acceptance for energy automata.  We first show that it is
sufficient to consider energy automata $( S, T)$ with precisely one
transition $( s, f, s')\in T$ for each pair of states $s, s'\in S$.
This will allow us to consider $T$ as a matrix $S\times S\to \E$ (as is
standard in weighted-automata theory~\cite{book/DrosteKV09}).

\begin{lemma}
  \label{le:onetransition}
  Let $\E'\subseteq \E$ be a subalgebra and $( S, T)$ an
  $\E'$-automaton.  There exists an $\E'$-automaton $( S, T')$ for which
  $\preach_{ \E'}( S, T)= \preach_{ \E'}( S, T')$ and $\pbuchi_{ \E'}(
  S, T)= \pbuchi_{ \E'}( S, T')$, and in which there is precisely one
  transition $( s, f, s')\in T'$ for all $s, s'\in S$.
\end{lemma}

Hence we may, without loss of generality, view the transitions $T$ of an
energy automaton as a matrix $T: S\times S\to \E$.  We can also let
$S=\{ 1,\dots, n\}$ and assume that the set of accepting states is $F=\{
1,\dots, k\}$ for $k\le n$.  Further, we can represent an initial state
$s_0\in S$ by the $s_0$th unit (column) vector $I^{ s_0}\in\{ \bbot,
\id\}^n$, defined by $I^{ s_0}_i= \id$ iff $i= s_0$, and $F$ by the
(column) vector $F^{ \le k}\in\{ \bbot, \id\}^n$ given by $F^{ \le k}_i=
\id$ iff $i\le k$.  Note that $T\in \E^{ n\times n}$ is an $n\times
n$-matrix of energy functions; as composition of energy functions is
written right-to-left, $T_{ ij}\in \E$ is the function on the transition
from $s_j$ to $s_i$.

\begin{theorem}
  \label{th:reach}
  Let $\E'\subseteq \E$ be a subalgebra.  For any $\E'$-automaton $(S,
  T)$ with $S=\{ 1,\dots, n\}$, $F=\{ 1,\dots, k\}$, $k\le n$, $s_0\le
  n$, and $x_0\in \Realnn$, we have $\preach_{ \E'}( S, T)( F, s_0,
  x_0)= \TRUE$ if, and only if, $\transpose F^{ \le k} T^* I^{ s_0}(
  x_0)\ne \bot$.
\end{theorem}

\begin{proof}
  Here $\transpose F^{ \le k}$ denotes the transpose of $F^{ \le k}$.
  By~\eqref{eq:starcontM}, we have $\transpose F^{ \le k} T^* I^{ s_0}=
  \sup_n( \transpose F^{ \le k} T^n I^{ s_0})$, so that $\transpose F^{
    \le k} T^* I^{ s_0}( x_0)\ne \bot$ iff $\transpose F^{ \le k} T^n
  I^{ s_0}( x_0)\ne \bot$ for some $n\in \Nat$, \ie~iff there is a
  finite run from $( s_0, x_0)$ which ends in a state in $F$. \qed
\end{proof}

Referring back to the example automaton $( S, T)$ from
Fig.~\ref{fi:eauto}, we display in Fig.~\ref{fi:eautoclos} the automaton
with transition matrix $T^*$.

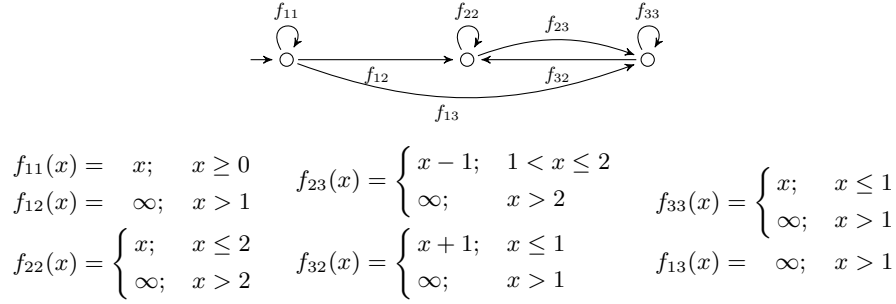
\begin{figure}[tbp]
  \centering
  \begin{tikzpicture}[->,>=stealth',shorten >=1pt,auto,node
    distance=2.0cm,initial text=,scale=.8,transform shape]
    \tikzstyle{every node}=[font=\small]
    \tikzstyle{every state}=[fill=white,shape=circle,inner
    sep=.5mm,minimum size=2.2mm,outer sep=.8mm]
    \node[state,initial] (1) at (0,0) {};
    \node[state] (2) at (3,0) {};
    \node[state] (3) at (6,0) {};
    \path (1) edge[out=120,in=60,loop] node[above] {$f_{ 11}$} (1);
    \path (1) edge node[below] {$f_{ 12}$} (2);
    \path (1) edge[out=-20,in=200] node[below left] {$f_{ 13}$} (3);
    \path (2) edge[out=120,in=60,loop] node[above] {$f_{ 22}$} (2);
    \path (2) edge[out=20,in=160] node[above] {$f_{ 23}$} (3);
    \path (3) edge node[below] {$f_{ 32}$} (2);
    \path (3) edge[out=120,in=60,loop] node[above] {$f_{ 33}$} (3);
  \end{tikzpicture}

  \vspace{-1ex}
  {\small
    \begin{minipage}{.3\linewidth}
      \begin{align*}
        f_{ 11}( x) &= \hspace*{.9em} x; \hspace*{1.4em} x\ge 0 \\
        f_{ 12}( x) &= \hspace*{.9em} \infty; \hspace*{1.0em} x> 1 \\
        f_{ 22}( x) &= \left\{
          \begin{aligned}
            &x;\hspace*{.4em} && x\le 2 \\
            &\infty; && x> 2
          \end{aligned}
        \right.
      \end{align*}
    \end{minipage}
    \hfill
    \begin{minipage}{.3\linewidth}
      \begin{align*}
        f_{ 23}( x) &= \left\{
          \begin{aligned}
            &x- 1; && 1< x\le 2 \\
            &\infty; && x> 2
          \end{aligned}
        \right. \\
        f_{ 32}( x) &= \left\{
          \begin{aligned}
            &x+ 1; && x\le 1 \\
            &\infty; && x> 1
          \end{aligned}
        \right.
      \end{align*}
    \end{minipage}
    \hfill
    \begin{minipage}{.3\linewidth}
      \begin{align*}
        f_{ 33}( x) &= \left\{
          \begin{aligned}
            &x;\hspace*{.4em} && x\le 1 \\
            &\infty; && x> 1
          \end{aligned}
        \right. \\
        f_{ 13}( x) &= \hspace*{.9em} \infty; \hspace*{1.0em} x> 1
      \end{align*}
    \end{minipage}}
  \caption{%
    \label{fi:eautoclos}
    The closure of the automaton from Fig.~\ref{fi:eauto}.}
\end{figure}

\begin{theorem}
  \label{th:buchi2}
  Let $\E'\subseteq \E$ be a subalgebra.  For any $\E'$-automaton $(S,
  T)$ with $S=\{ 1,\dots, n\}$, $F=\{ 1,\dots, k\}$, $k\le n$, $s_0\le
  n$, and $x_0\in \Realnn$, we have $\pbuchi_{ \E'}( S, T)( F, s_0,
  x_0)= T^{ \omega_k} I^{ s_0}( x_0)$.
\end{theorem}

\begin{proof}%[of Theorem~\ref{th:buchi2}]
  This is a standard result for \emph{complete} semiring-semimodule
  pairs, \cf~\cite{EK2}.  Now $( \E, \V)$ is not complete, but the
  properties developed in the previous section allow us to show the
  result nevertheless.  We need to see that for all $M \in \E^{n \times
    n}$ and $1 \leq i \leq n$,
  \begin{equation*} 
    (M^\omega)_{i}=  \sup\{\cdots M_{k_3,k_2}M_{k_2,k_1}M_{k_1,i} : 1
    \leq k_1,k_2,\ldots \leq n \}\,,
  \end{equation*}
  which we shall deduce inductively from~\eqref{eq:omegaM}.

  Let $a\in \E^{ \ell\times \ell}$, $d\in \E^{ m\times m}$, for $\ell+
  m= n$, and let $i\in\{ 1,\dots, \ell\}$.  Then the $i$th component of
  $M^\omega$ is the $i$th component of $(a \vee bd^*c)^\omega \vee
  d^\omega c (a \vee bd^*c)^*$.  By induction hypothesis, the $i$th
  component of $(a \vee bd^*c)^\omega$ is the supremum of all infinite
  products $(\cdots M_{k_2,k_1}M_{k_1,i})$ such that $1 \leq k_j \leq m$
  for an infinite number of indices $j$, and similarly, the $i$th
  component of $d^\omega c (a \vee bd^*c)^*$ is the supremum of all
  infinite products $(\cdots M_{k_2,k_1}M_{k_1,i})$ such that $1 \leq
  k_j \leq m$ for a finite number of indices $j$.  Thus, the $i$th
  component of $(a \vee bd^*c)^\omega \vee d^\omega c (a \vee bd^*c)^*$
  is the supremum of all infinite products $(\cdots
  M_{k_2,k_1}M_{k_1,i})$.  \qed
\end{proof}

We remark that our decision algorithms are \emph{static} in the sense
that the matrix expressions can be pre-computed and then re-used to
decide reachability and B{\"u}chi acceptance for different values $x_0$
of initial energies.

Using elementary reasoning on infinite paths, we can provide an
alternative characterization of B{\"u}chi acceptance which does not use
the omega operations:

\begin{theorem}
  \label{th:buchi}
  Let $\E'\subseteq \E$ be a subalgebra.  For any $\E'$-automaton $(S,
  T)$ with $S=\{ 1,\dots, n\}$, $F=\{ 1,\dots, k\}$, $k\le n$, $s_0\le
  n$, and $x_0\in \Realnn$, we have $\pbuchi_{ \E'}( S, T)( F, s_0,
  x_0)= \TRUE$ if, and only if, there exists $j\le k$ for which
  \begin{equation*}
    \transpose I^j T T^* I^j \, \transpose I^j T^* I^{ s_0}( x_0)\ge
    \transpose I^j T^* I^{ s_0}( x_0)\ne \bot.
  \end{equation*}
\end{theorem}

\begin{corollary}
  \label{co:decidereach}
  \label{co:decidebuchi}
  For subalgebras $\E'\subseteq \E$ of computable functions in which it
  is decidable for each $f\in \E'$ whether $f( x)\le x$,
  Problems~\ref{pb:reach} and~\ref{pb:buchi} are decidable.  For an
  energy automaton with $n$ states and $m$ transitions, the decision
  procedures use $O( m+ n^3)$, respectively $O( m+ n^4)$, algebra
  operations.
\end{corollary}

\begin{proof}
  Maxima and compositions of computable functions are again computable,
  and if it is decidable for each $f\in \E'$ whether $f( x)\le x$, then
  also $f^*$ is computable for each $f\in \E'$.  Hence all matrix
  operations used in Lemma~\ref{le:onetransition} and
  Theorems~\ref{th:reach} and~\ref{th:buchi} are computable.  The number
  of operations necessary in the construction in the proof of
  Lemma~\ref{le:onetransition} is $O( m)$, and, using \eg~the
  Floyd-Warshall algorithm to compute $T^*$, $O( n^3)$ operations are
  necessary to compute $\transpose I^{ \le k} T^* I^{ s_0}$.  \qed
\end{proof}

We proceed to identify two important subclasses of computable energy
functions, which cover most of the related work mentioned in the
introduction, and to give complexity results on their reachability and
B{\"u}chi acceptance problems.

The \emph{integer update functions} in $\E$ are the functions $f_k$,
for $k\in \Int$, given by
\begin{equation*}
  f_k( x)= \left\{
  \begin{array}{cl}
    x+ k &\quad\text{if } x\ge \max( 0, -k)\,, \\
    \bot &\quad\text{otherwise}\,,
  \end{array} \right.
\end{equation*}
together with $f_\infty:= \ttop$.  These are the update functions
usually considered in integer-weighted automata and
VASS~\cite{conf/ictac/FahrenbergJLS11, DBLP:conf/fsttcs/ChatterjeeDHR10,
  DBLP:conf/icalp/ChatterjeeD10, DBLP:conf/icalp/BrazdilJK10,
  DBLP:conf/rp/Chaloupka10, DBLP:journals/ipl/Chan88,
  DBLP:conf/formats/BouyerFLMS08}.  We have $f_\ell f_k= f_{ k+ \ell}$
and $f_k\vee f_\ell= f_{ \max( k, \ell)}$, and $f_k^*= f_0$ for $k\le 0$
and $f_k^*= f_\infty$ for $k> 0$, whence the class $\Eint$ of integer
update functions forms a subalgebra of $\E$.  A function $f_k\in \Eint$
can be represented by the integer $k$, and algebra operations can then
be performed in constant time.  Hence Corollary~\ref{co:decidereach}
implies the following result.

\begin{theorem}
  For $\Eint$-automata, Problems~\ref{pb:reach} and~\ref{pb:buchi} are
  decidable in PTIME.
\end{theorem}

Next we turn our attention to piecewise affine functions as used
in~Fig.~\ref{fi:example}.

\begin{definition}
  \label{de:fpwint}
  A function $f\in \E$ is said to be \emph{(rational) piecewise affine}
  if there exist $x_0< x_1<\dots< x_k\in \Rat$ such that $f( x)\ne \bot$
  iff $x\ge x_0$ or $x> x_0$, $f( x_j)\in \Rat\cup\{ \bot\}$ for all
  $j$, and all restrictions $f\rest{ \mathopen] x_j, x_{ j+
      1}\mathclose[}$ and $f\rest{ \mathopen] x_k, \infty\mathclose[}$
  are affine functions $x\mapsto a_jx+ b_j$ with $a_j, b_j\in \Rat$,
  $a_j\ge 1$.
\end{definition}

\begin{figure}[tbp]
  \centering
  \begin{tikzpicture}[radius=2pt]
    \path (-.5,0) edge[black!30] (5,0);
    \path (0,-.5) edge[black!30] (0,5);
    \foreach \x in {1,2,3,4,5} \node at (\x,-.4) {\small \color{gray} $\x$};
    \foreach \y in {1,2,3,4,5} \node at (-.4,\y) {\small \color{gray} $\y$};
    \fill (2,.5) circle;
    \draw (2,.5) -- (3,2) circle;
    \fill (3,2.3) circle;
    \draw (3,2.7) circle -- (4.5,4.2) circle;
    \fill (4.5,4.5) circle;
    \draw (4.5,4.5) -- (4.75,5);
    \draw (4.75,5) edge[dotted] (5,5.5);
    \node at (8,2.5) {$%
      f( x)=
      \begin{cases}
        .5 &( x= 2) \\
        1.5\, x- 2.5 &( 2< x< 3) \\
        2.3 &( x= 3) \\
        x- .3 &( 3< x< 4.5) \\
        4.5 &( x= 4.5) \\
        2\, x- 4.5 &( x> 4.5)
      \end{cases}
      $};
  \end{tikzpicture}
  \caption{%
    \label{fi:example-pw}
    A piecewise affine energy function}
\end{figure}
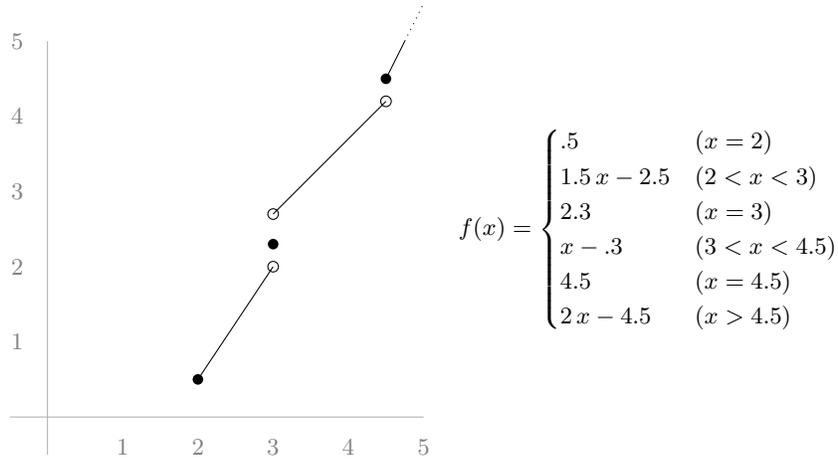

Note that the definition does not make any assertion about continuity at
the $x_j$, but~\eqref{eq:deriv1} implies that $\lim_{ x\nearrow x_j} f(
x)\le f( x_j)\le \lim_{ x\searrow x_j} f( x)$.  A piecewise affine
function as above can be represented by its break points $x_0,\dotsc,
x_k$, the values $f( x_0),\dotsc, f( x_k)$, and the numbers $a_0,
b_0,\dotsc, a_k, b_k$.  These functions arise in the reduction used
in~\cite{DBLP:conf/hybrid/BouyerFLM10} to show decidability of energy
problems for one-clock timed automata with transition updates.  The
notion of \emph{integer piecewise affine} functions is defined
similarly, with all occurrences of $\Rat$ above replaced by $\Int$.
Fig.~\ref{fi:example-pw} shows an example of a piecewise affine
function.

The class $\Epw$ of piecewise affine energy functions forms a subsemiring
of $\E$: if $f, g\in \Epw$ with break points $x_0,\dots, x_k$,
$y_0,\dots, y_\ell$, respectively, then $f\vee g$ is piecewise affine
with break points a subset of $\{ x_0,\dots, x_k, y_0,\dots, y_\ell\}$,
and $gf$ is piecewise affine with break points a subset of $\{
x_0,\dots, x_k,$ $f^{ -1}( y_0),\dots, f^{ -1}( y_\ell)\}$.

Let, for any $k\in \Rat$, $g_k^-, g_k^+:[ 0, \infty]_\bot\to[][ 0,
\infty]_\bot$ be the functions defined by
\begin{equation*}
  g_k^-( x)= \left\{
  \begin{array}{cl}
    x &\quad\text{for } x< k\,, \\
    \infty &\quad\text{for } x\ge k\,,
  \end{array} \right.
  \qquad
  g_k^+( x)= \left\{
  \begin{array}{cl}
    x &\quad\text{for } x\le k\,, \\
    \infty &\quad\text{for } x> k\,.
  \end{array} \right.
\end{equation*}
By Lemma~\ref{le:prop-star} (and noticing that for all $f\in \Epw$,
$\sup\{ x\mid f( x)\le x\}$ is rational), $\Epw$ completed with the
functions $g_k^-$, $g_k^+$ forms a subalgebra of $\E$.

Remark that, unlike $\Epw$, the class $\Epwi$ of integer piecewise affine
functions does \emph{not} form a subsemiring of $\E$, as composites of
$\Epwi$-functions are not necessarily integer piecewise affine.
As an example, for the functions $f, g\in \Epwi$ given by
\begin{equation*}
  f( x)= 2x\,, \qquad g( x)=
  \begin{cases}
    x+ 1; & x< 3\,, \\
    x+ 2; & x\ge 3\,,
  \end{cases}
\end{equation*}
we have
\begin{equation*}
  g( f( x))=
  \begin{cases}
    2x+ 1; & x< 1.5\,, \\
    2x+ 2; & x\ge 1.5\,.
  \end{cases}
\end{equation*}
which is not integer piecewise affine.  Similarly, the class of rational
\emph{affine} functions $x\mapsto ax+ b$ (without break points) is not
closed under maximum, and $\Epw$ is the semiring generated by rational
affine functions.

\begin{theorem}
  \label{th:exptime}
  For $\Epw$-automata, Problems~\ref{pb:reach} and~\ref{pb:buchi} are
  decidable in EXPTIME.
\end{theorem}

\begin{proof}%[of Theorem~\ref{th:exptime}]
  We need to show that it is decidable for each $f\in \Epw$ whether $f(
  x)\le x$.  Let thus $f$ be a piecewise affine function, with
  representation $( x_0,\dots, x_k, f( x_0),$ $\ldots, f( x_k),
  a_0,\dots, a_k, b_0,\dots, b_k)$.  If $x< x_0$, then $f( x)= \bot\le
  x$.  If $x= x_j$ for some $j$, we can simply compare $x_j$ with $f(
  x_j)$.

  Assume now that $x\in \mathopen] x_j, x_{ j+ 1}\mathclose[$ for some
  $j$.  If $a_j x_j+ b_j\le x_j$ and $a_j x_{ j+ 1}+ b_j\le x_{ j+ 1}$,
  then also $f( x)\le x$ by~\eqref{eq:deriv1}.  Likewise, if $a_j x_j+
  b_j> x_j$ and $a_j x_{ j+ 1}+ b_j> x_{ j+ 1}$, then also $f( x)> x$.
  The case $a_j x_j+ b_j> x_j$, $a_j x_{ j+ 1}+ b_j\le x_{ j+ 1}$ cannot
  occur because of~\eqref{eq:deriv1}, and if $a_j x_j+ b_j\le x_j$ and
  $a_j x_{ j+ 1}+ b_j> x_{ j+ 1}$, then $a_j> 1$, and $f( x)\le x$ iff
  $x\le \frac{ b_j}{ 1- a_j}$.

  For the case $x\in \mathopen] x_k, \infty\mathclose[$, the arguments
  are similar: if $a_k x_k+ b_k> x_k$, then also $f( x)> x$; if $a_k
  x_k+ b_k\le x_k$ and $a_k= 1$, then also $f( x)\le x$, and if $a_k> 1$
  in this case, then $f( x)\le x$ iff $x\le \frac{ b_k}{ 1- a_k}$.

  Using Corollary~\ref{co:decidereach}, we have hence shown
  decidability.  For the complexity claim, we note that all algebra
  operations in $\Epw$ can be performed in time linear in the size of
  the representations of the involved functions.  However, the maximum
  and composition operations may double the size of the representations,
  hence our procedure may take time $O( 2^{ m+ n^3} p)$ for
  reachability, and $O( 2^{ m+ n^4} p)$ for B{\"u}chi acceptance, for an
  $\Epw$-automaton with $n$ states, $m$ transitions, and energy
  functions of representation length at most $p$.  \qed
\end{proof}

In the setting of $\Epw$-automata and their application to one-clock
weighted timed automata with transition updates, our
Theorem~\ref{th:buchi} is a generalization of~\cite[Lemmas~24,
25]{DBLP:conf/hybrid/BouyerFLM10}.  Complexity of the decision procedure
was left open in~\cite{DBLP:conf/hybrid/BouyerFLM10}; as the conversion
of a one-clock weighted timed automaton to an $\Epw$-automaton incurs an
exponential blowup, we now see that their procedure is
doubly-exponential.

Considerations similar to the above show that also the setting of
\emph{piecewise polynomial} energy functions allows an application of
Theorem~\ref{th:buchi} to show energy problems on the exponentially
weighted timed automata from~\cite{DBLP:conf/hybrid/BouyerFLM10}
decidable.
%
% On the other hand, we can show a PSPACE lower bound already for
% $\Epwi$-automata:

% \begin{theorem}
%   \label{th:pspacehard}
%   For $\Epwi$-automata, Problems~\ref{pb:reach} and~\ref{pb:buchi} are
%   PSPACE-hard.
% \end{theorem}

% \begin{proof}[sketch]
%   The idea for this proof comes from the proof of~\cite[Thm.~2]{CY92}.
%   Given that reachability can be reduced to B{\"u}chi acceptance, we
%   only have to show PSPACE-hardness of the reachability problem.  This
%   in turn is shown by a reduction from the acceptance problem for
%   deterministic linear bounded automata (LBA).  For a LBA $M$, which
%   uses space $n$ equal to the length of the input, and an input $x$, we
%   construct an energy automaton $(S,T)$ with integer piecewise affine
%   energy functions with two distinguished states $s_0$ and $s_F$ and
%   initial value $y=0$ such that $M$ accepts $x$ if, and only if, there
%   exists a finite run of $(S,T)$ from $(s_0,y)$ to $s_F$.  The details
%   of the proof are in~\cite{arxiv/EsikFLQ13}.  \qed
% \end{proof}

\section{Multi-dimensional Energy Automata and Games}
\label{se:multi}
\label{se:reachgamemulti}

Next we turn our attention to several variants of energy automata. 
We will generally stick to the set $\Epwi$ of integer piecewise affine
energy functions; the fact that $\Epwi$ is not a subsemiring of $\E$
will not bother us here.

An \emph{$n$-dimensional integer piecewise affine energy automaton}, or
\emph{$\Epwi^n$-auto\-maton} for short, $( S, T)$, for $n\in \Natp$,
consists of finite sets $S$ of states and $T\subseteq S\times
\Epwi^n\times S$ of transitions.  A \emph{global state} in such an
automaton is a pair $(s,\vec{x})\in S\times \Nat^n$, and
\emph{transitions} are of the form
$(s,\vec{x})\to[ \vec{f}](s',\vec{x}')$
such that $(s,\vec{f},s')\in T$ and $\vec{x}'( i)=\vec{f}( i)(\vec{x}(
i))$ for each $i\in\{1,\dots,n\}$.

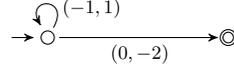
\begin{figure}[tbp]
  \centering
  \begin{tikzpicture}[->,>=stealth',shorten >=1pt,auto,node
    distance=2.0cm,initial text=,scale=.8,transform shape]
    \tikzstyle{every node}=[font=\small]
    \tikzstyle{every state}=[fill=white,shape=circle,inner
    sep=.5mm,minimum size=2.2mm,outer sep=.8mm]
    \node[state,initial] (1) at (0,0) {};
    \node[state,accepting] (2) at (3,0) {};
    \path (1) edge[out=60,in=120,loop] node[right,pos=.3] {$( -1, 1)$} (1);
    \path (1) edge node[below] {$( 0, -2)$} (2);
  \end{tikzpicture}
  \caption{%
    \label{fi:2dim-nokleene}
    A simple two-dimensional VASS}
\end{figure}

For reachability in $\Epwi^n$-automata (with $n\ge 2$), our algebraic
results do not apply.  To see this, we refer to the reachability problem
in Fig.~\ref{fi:2dim-nokleene}: with initial energy $( 1, 1)$, the loop
needs to be taken precisely once, but with initial energy $( 2, 0)$, one
needs to loop twice.  Hence there is no static algorithm which can
decide reachability for this VASS.

However, we remark that $\Epwi^n$-automata are \emph{well-structured
  transition systems}~\cite{DBLP:journals/tcs/FinkelS01}, with ordering
on global states defined by $(s,\vec{x})\preceq (s',\vec{x}')$ iff
$s=s'$ and $\vec{x}(i)\leq\vec{x}'(i)$ for each $i= 1,\dots, n$ (here we
also have to assume $x_0\in \Nat$).  Also, the reachability problem for
energy automata is a \emph{control state reachability problem} in the
sense of~\cite{DBLP:journals/iandc/AbdullaCJT00}.  Decidability of the
reachability problem for $\Epwi^n$-automata thus follows from the
decidability of the control state reachability problem for
well-structured transition
systems~\cite{DBLP:journals/iandc/AbdullaCJT00}.  Note that B{\"u}chi
acceptance is not generally decidable for well-structured transition
systems (it is undecidable for lossy counter
machines~\cite{DBLP:conf/rp/Schnoebelen10}), so our reduction proof does
not imply a similar result for B{\"u}chi acceptance.

\begin{theorem}
  \label{thm_multiple_var_reach}
  The reachability problem for $\Epwi^n$-automata with $x_0\in \Nat$ is
  decidable.
\end{theorem}

Next we show that if the requirement~\eqref{eq:deriv1} on energy
functions, that $f( x_2)\ge f( x_1)+ x_2- x_1$ for each $x_1\le x_2$, is
lifted, then reachability becomes undecidable from dimension $4$.  We
call such functions \emph{flat} energy functions; remark that we still
require them to be strictly increasing, but the derivative, where it
exists, may be less than $1$.  The class of all flat energy functions is
denoted $\Eflat$ and its restrictions by $\Eflatpw$, $\Eflatpwi$.

\begin{theorem}
  \label{th:undec4dim}
  The reachability problem for $\Eflatpw^4$-automata is undecidable.
\end{theorem}

Next we extend our energy automata formalism to (turn based)
\emph{reachability games}.  Let $(S,T)$ be an $n$-dimensional energy
automaton such that $S= S_A\cup S_B$ forms a partition of $S$ and
$T\subseteq (S_A\times\Epwi^n\times S_B)\cup (S_B\times\Epwi^n\times
S_A)$.  Then $(S,S_A,S_B,T)$ induces an $n$-dimensional \emph{energy
  game} $G$.  The intuition of the reachability game is that the two
players $A$ and $B$ take turns to move along the game graph $( S, T)$,
updating energy values at each turn.  The goal of player~$A$ is to reach
a state in $F$, the goal of player~$B$ is to prevent this from
happening.

The reachability game is a \emph{coverability game} in the sense
of~\cite{DBLP:journals/entcs/RaskinSB05}.  In general, the reachability
game on well-structured transition systems is
undecidable~\cite{DBLP:journals/logcom/AbdullaBd08}.  Indeed, the games
on VASS considered in~\cite{DBLP:conf/icalp/BrazdilJK10} are a special
case of reachability games on energy automata with integer update
functions; their undecidability is shown
in~\cite{DBLP:journals/logcom/AbdullaBd08,
  DBLP:conf/icalp/BrazdilJK10}. It is hence clear that it is undecidable
whether player~$A$ wins the reachability game in $2$-dimensional
$\Eint$-automata.  As a corollary, we can show that for \emph{flat}
energy functions, already one-dimensional reachability games are
undecidable.

\begin{theorem}
  \label{th:undec2game}
  Whether player~$A$ wins the reachability game in $\Eint^2$-automata is
  undecidable.
\end{theorem}

\begin{theorem}
  \label{th:flatgame}
  It is undecidable for $\Eflatpw$-automata whether player~$A$ wins the
  reachability game.
\end{theorem}

\begin{proof}[sketch]
  The proof is by reduction from reachability games on $2$-dimen\-sional
  $\Eint$-automata to reachability games on $1$-dimensional
  $\Eflatpwi$-automata.  The intuition is that the new energy variable
  $x$ encodes the two old ones as $x= 2^{ x_1} 3^{ x_2}$, and then
  transitions in the $2$-dimensional game are encoded using gadgets in
  which the other player may interrupt to demand proof that the required
  inequalities for $x_1$ \emph{and} $x_2$ were satisfied.  The energy
  functions in the so-constructed $1$-dimensional automaton are
  piecewise affine because the original ones were integer updates.  The
  details of the proof are in appendix. \qed
\end{proof}

\section{Conclusion}

We have in this paper introduced a functional framework for modeling and
analyzing energy problems.  We have seen that our framework encompasses
most existing formal approaches to energy problems, and that it allows
an application of the theory of automata over semirings and semimodules
to solve reachability and B{\"u}chi acceptance problems in a generic
way.  For the important class of piecewise affine energy functions, we
have shown that reachability and B{\"u}chi acceptance are PSPACE-hard
and decidable in EXPTIME.  As our algorithm is static, computations do
not have to be repeated in case the initial energy changes.  Also,
decidability of B{\"u}chi acceptance implies that LTL model checking is
decidable for energy automata.

In the last part of this paper, we have seen that one quickly comes into
trouble with undecidability if the class of energy functions is extended
or if two-player games are considered.  This can be remedied by
considering \emph{approximate} solutions instead, using notions of
distances for energy automata akin to the ones
in~\cite{DBLP:conf/fsttcs/FahrenbergLT11} to provide quantitative
measures for similar energy behavior.
 
Another issue that remains to be investigated is reachability and
B{\"u}chi problems for one-dimensional energy automata with flat energy
functions; we plan to do this in future work.

%\bibliography{eabib}

\begin{thebibliography}{10}

\bibitem{DBLP:journals/logcom/AbdullaBd08}
P.~A. Abdulla, A.~Bouajjani, and J.~d'Orso.
\newblock Monotonic and downward closed games.
\newblock {\em J. Log. Comput.}, 18(1):153--169, 2008.

\bibitem{DBLP:journals/iandc/AbdullaCJT00}
P.~A. Abdulla, K.~{\v C}er{\=a}ns, B.~Jonsson, and Y.-K. Tsay.
\newblock Algorithmic analysis of programs with well quasi-ordered domains.
\newblock {\em Inf. Comput.}, 160(1-2):109--127, 2000.

\bibitem{DBLP:journals/tcs/AlurD94}
R.~Alur and D.~L. Dill.
\newblock A theory of timed automata.
\newblock {\em {Theor. Comput. Sci.}}, 126(2):183--235, 1994.

\bibitem{DBLP:conf/hybrid/AlurTP01}
R.~Alur, S.~L. Torre, and G.~J. Pappas.
\newblock Optimal paths in weighted timed automata.
\newblock In {\em HSCC}, pp. 49--62, 2001.

\bibitem{DBLP:conf/hybrid/BehrmannFHLPRV01}
G.~Behrmann, A.~Fehnker, T.~Hune, K.~G. Larsen, P.~Pettersson, J.~Romijn, and
  F.~W. Vaandrager.
\newblock Minimum-cost reachability for priced timed automata.
\newblock In {\em HSCC}, pp. 147--161, 2001.

\bibitem{BEbook}
S.~L. Bloom and Z.~{\'E}sik.
\newblock {\em Iteration Theories: The Equational Logic of Iterative
  Processes}.
\newblock EATCS monographs on theoretical computer science. {Springer}, 1993.

\bibitem{DBLP:conf/hybrid/BouyerFLM10}
P.~Bouyer, U.~Fahrenberg, K.~G. Larsen, and N.~Markey.
\newblock Timed automata with observers under energy constraints.
\newblock In {\em HSCC}, pp. 61--70, 2010.

\bibitem{DBLP:conf/formats/BouyerFLMS08}
P.~Bouyer, U.~Fahrenberg, K.~G. Larsen, N.~Markey, and J.~Srba.
\newblock Infinite runs in weighted timed automata with energy constraints.
\newblock In {\em FORMATS}, pp. 33--47, 2008.

\bibitem{DBLP:conf/qest/BouyerLM12}
P.~Bouyer, K.~G. Larsen, and N.~Markey.
\newblock Lower-bound constrained runs in weighted timed automata.
\newblock In {\em QEST}, pp. 128--137, 2012.

\bibitem{DBLP:conf/icalp/BrazdilJK10}
T.~Br{\'a}zdil, P.~Jan{\v c}ar, and A.~Ku{\v c}era.
\newblock Reachability games on extended vector addition systems with states.
\newblock In {\em ICALP}, pp. 478--489, 2010.

\bibitem{DBLP:conf/rp/Chaloupka10}
J.~Chaloupka.
\newblock Z-reachability problem for games on 2-dimensional vector addition
  systems with states is in {P}.
\newblock In {\em RP}, pp. 104--119, 2010.

\bibitem{DBLP:journals/ipl/Chan88}
T.-h. Chan.
\newblock The boundedness problem for three-dimensional vector addition systems
  with states.
\newblock {\em {Inf. Proc. Letters}}, 26(6):287--289, 1988.

\bibitem{DBLP:conf/icalp/ChatterjeeD10}
K.~Chatterjee and L.~Doyen.
\newblock Energy parity games.
\newblock In {\em ICALP}, pp. 599--610, 2010.

\bibitem{DBLP:conf/fsttcs/ChatterjeeDHR10}
K.~Chatterjee, L.~Doyen, T.~A. Henzinger, and J.-F. Raskin.
\newblock Generalized mean-payoff and energy games.
\newblock In {\em FSTTCS}, pp. 505--516, 2010.

\bibitem{book/Conway71}
J.~H. Conway.
\newblock {\em Regular Algebra and Finite Machines}.
\newblock Chapman and Hall, 1971.

\bibitem{DBLP:conf/csl/DegorreDGRT10}
A.~Degorre, L.~Doyen, R.~Gentilini, J.-F. Raskin, and S.~Torunczyk.
\newblock Energy and mean-payoff games with imperfect information.
\newblock In {\em CSL}, pp. 260--274, 2010.

\bibitem{book/DrosteKV09}
M.~Droste, W.~Kuich, and H.~Vogler.
\newblock {\em Handbook of Weighted Automata}.
\newblock {Springer}, 2009.

\bibitem{DBLP:conf/mfcs/EsikK04}
Z.~{\'E}sik and W.~Kuich.
\newblock An algebraic generalization of omega-regular languages.
\newblock In {\em MFCS}, pp. 648--659, 2004.

\bibitem{EK2}
Z.~{\'E}sik and W.~Kuich.
\newblock A semiring-semimodule generalization of $\omega$-regular languages,
  {P}arts 1 and 2.
\newblock {\em J. Aut. Lang. Comb.}, 10:203--264, 2005.

\bibitem{conf/ictac/FahrenbergJLS11}
U.~Fahrenberg, L.~Juhl, K.~G. Larsen, and J.~Srba.
\newblock Energy games in multiweighted automata.
\newblock In {\em ICTAC}, pp. 95--115, 2011.

\bibitem{DBLP:conf/fsttcs/FahrenbergLT11}
U.~Fahrenberg, A.~Legay, and C.~Thrane.
\newblock The quantitative linear-time--branching-time spectrum.
\newblock In {\em FSTTCS}, pp. 103--114, 2011.

\bibitem{DBLP:journals/tcs/FinkelS01}
A.~Finkel and P.~Schnoebelen.
\newblock Well-structured transition systems everywhere!
\newblock {\em {Theor. Comput. Sci.}}, 256(1-2):63--92, 2001.

\bibitem{DBLP:conf/mfcs/Kozen90}
D.~Kozen.
\newblock On {K}leene algebras and closed semirings.
\newblock In {\em MFCS}, pp. 26--47, 1990.

\bibitem{DBLP:journals/iandc/Kozen94}
D.~Kozen.
\newblock A completeness theorem for kleene algebras and the algebra of regular
  events.
\newblock {\em Inf. Comput.}, 110(2):366--390, 1994.

\bibitem{DBLP:conf/RelMiCS/MathieuD05}
V.~Mathieu and J.~Desharnais.
\newblock Verification of pushdown systems using omega algebra with domain.
\newblock In {\em RelMiCS}, pp. 188--199, 2005.

\bibitem{book/PerrinP04}
D.~Perrin and J.-E. Pin.
\newblock {\em Infinite Words: Automata, Semigroups, Logic and Games}.
\newblock Academic Press, 2004.

\bibitem{DBLP:conf/lata/Quaas11}
K.~Quaas.
\newblock On the interval-bound problem for weighted timed automata.
\newblock In {\em LATA}, pp. 452--464, 2011.

\bibitem{DBLP:journals/entcs/RaskinSB05}
J.-F. Raskin, M.~Samuelides, and L.~V. Begin.
\newblock Games for counting abstractions.
\newblock {\em {Electr. Notes Theor. Comput. Sci.}}, 128(6):69--85, 2005.

\bibitem{DBLP:conf/rp/Schnoebelen10}
P.~Schnoebelen.
\newblock Lossy counter machines decidability cheat sheet.
\newblock In {\em RP}, pp. 51--75, 2010.

\bibitem{DBLP:conf/icalp/Wilke91}
T.~Wilke.
\newblock An {E}ilenberg theorem for infinity-languages.
\newblock In {\em ICALP}, pp. 588--599, 1991.

\end{thebibliography}

\newpage
\appendix

\section*{Appendix: Proofs}

\newenvironment{proofapp}[1][]{%
  \par\medskip\noindent\textbf{Proof #1.}\quad}{}

\begin{proofapp}[of Lemma~\ref{le:distr}]
  The pointwise supremum of any set of extended energy functions is an
  extended energy function. Indeed, if $f_i$, $i \in I$ are extended
  energy functions and $x < y$ in $\Realnn$, then $f_i(y) \geq f_i(x) + y -
  x$ for all $i$. It follows that $\sup_{i\in I}f_i(y) \geq \sup_{i \in
    I} f_i(x) + y - x$.  Also, since $f_i(\bot) = \bot$ for all $i \in
  I$, $\sup_{i \in I}f_i(\bot) = \bot$.  Finally, if there is some $i$
  such that $f_i(\infty) = \infty$, then $\sup_{i \in I}f_i(\infty) =
  \infty$.  Otherwise each function $f_i$ is constant with value $\bot$.

  The fact that $(\sup_{i \in I}f_i)h = \sup_{i \in I} f_i h$ is now
  clear, since the supremum is taken pointwise: For all $x$, $((\sup_{i
    \in I}f_i)h)( x)=( \sup_{ i\in I}f_i)( h( x))= \sup_{ i\in I}( f_i(
  h( x)))$ and also $( \sup_{ i\in I} f_i h)( x)= \sup_{ i\in I}( f_i(
  h( x)))$. \qed
\end{proofapp}

\begin{proofapp}[of Lemma~\ref{le:idmsring}]
  Recall first~\cite{book/DrosteKV09} that an idempotent semiring
  is an algebraic structure $( S, \oplus, \otimes, \mathbb0, \mathbb1)$
  satisfying, for all $a, b, c\in S$, the following axioms:
  \begin{gather}
    a\oplus( b\oplus c)=( a\oplus b)\oplus c \quad a\oplus b= b\oplus a
    \quad a\oplus \mathbb0= a \label{eq:dioid.plus} \\
    a\otimes( b\otimes c)=( a\otimes b)\otimes c \qquad \mathbb1\otimes
    a=
    a\otimes \mathbb1= a \label{eq:dioid.times} \\
    a\otimes( b\oplus c)=( a\otimes b)\oplus( a\otimes c) \qquad
    \mathbb0\otimes a= a\otimes
    \mathbb0= \mathbb0 \label{eq:dioid.distl} \\
    ( a\oplus b)\otimes c=( a\otimes c)\oplus( b\otimes c) \qquad
    a\oplus a= a \label{eq:dioid.distr}
  \end{gather}

  Now the axioms for $\vee$ in~\eqref{eq:dioid.plus} are clear: maximum is
  associative and commutative, with neutral element $\bbot$.  Similarly,
  \eqref{eq:dioid.times}~states that composition is associative with
  neutral element the identity function $\id$.  As
  for~\eqref{eq:dioid.distl}, left distributivity of $\circ$ over $\vee$
  follows from monotonicity of the functions in $\E$: we have $(
  h( f\vee g))( x)= h( \max( f( x), g( x)))$ and $( hf\vee h g)( x)=
  \max( h( f( x)), h( g( x)))$.  The fact that $\bbot$ is absorbing is
  clear.  Right distributivity in~\eqref{eq:dioid.distr} can be
  similarly shown (but does not need monotonicity), and $\vee$ is
  idempotent by definition.  $\le$ is the natural order on $\E$
  because $\vee$ is given pointwise. \qed
\end{proofapp}

\begin{proofapp}[of Lemma~\ref{le:prop-star}]
  It is clear that $f^*$ is an energy function for any $f\in \E$.  For
  the other claims, we first note that if $g\in \E$ is such that there
  is $k$ for which $g( x)= x$ for $x< k$ and $g( x)= \infty$ for $x> k$,
  then $g^*( x)= g( x)$ for $x\ne k$, and if $g( k)= k$ or $g( k)=
  \infty$, then also $g^*( k)= g( k)$.
  
  Now let $g\in \E$.  If there is $f\in \E$ with $g=
  f^*$, then we set $k= \sup\{ x\mid f( x)\le x\}$.  Then $f( x)> x$ and
  hence $g( x)= \infty$ for all $x> k$, and whenever $x< k$, then there
  is $y$ with $x\le y\le k$ and $f( y)\le y$, hence
  by~\eqref{eq:deriv1}, $f( x)\le x$, so that $g( x)= x$.  If $f( k)\le
  k$, then $g( k)= k$, otherwise $g( k)= \infty$ as claimed. \qed
\end{proofapp}

\begin{proofapp}[of Lemma~\ref{le:distl}]
  We have $g f^*( \bot)= \bot= \sup_n g f^n( \bot)$, and also $g f^*(
  \infty)= \infty= \sup_n g f^n( \infty)$.  Also, if $g= \bbot$, then $g
  f^*( x)= \bot= \sup_n g f^n( x)$ for all $x\in[ 0, \infty]_\bot$.

  We are left with showing $g f^*( x)= \sup_n g f^n( x)$ for all $x\in
  \Realnn$ in case $g\ne \bbot$.  Let thus $x\in \Realnn$.  If $f( x)\le
  x$, then also $f^n( x)\le x$ for all $n\in \Nat$, hence $g f^n( x)\le
  g( x)$ for all $n$, so that $\sup_n g f^n( x)= g( x)= g f^*( x)$.

  Now assume instead that $f( x)> x$, so that $f( x)- x= M> 0$.
  By~\eqref{eq:deriv1} applied to $x_1= x$ and $x_2= f( x)$, we have
  $f^2( x)\ge f( x)+ f( x)- x= f( x)+ M$, hence by induction, $f^{ n+
    1}( x)\ge f^n( x)+ M$ for all $n\in \Nat$.  Hence the sequence $(
  f^n( x))_{ n\in \Nat}$ increases without bound, so that there must be
  $N\in \Nat$ for which $g f^N( x)\ne \bot$.  Again
  using~\eqref{eq:deriv1}, we see that $g f^{ n+ 1}( x)\ge g f^n( x)+ M$
  for all $n\ge N$, so that also the sequence $( g f^n( x))_{ n\in
    \Nat}$ increases without bound, whence $\sup_n g f^n( x)= \infty= g
  f^*( x)$.  \qed
\end{proofapp}

\begin{proofapp}[of Proposition~\ref{pr:EcontKA}]
  Let $f, g, h\in \E$.  By Lemmas~\ref{le:distl}
  and~\ref{le:distr}, $g f^* h=( \sup_n g f^n) h= \sup_n( g f^n h)$. \qed
\end{proofapp}

\begin{proofapp}[of Lemma~\ref{le:semimod}]
  It is easily verified that $u( f g)=( u f) g$ and $u\circ \id= u$ for
  all $u \in \V$ and $f,g \in \E$.  Moreover,
  \begin{align*}
    ((u \vee v)f)(x) &= (u\vee v)(f(x))\\
    &= u(f(x)) \vee v(f(x))\\
    &= (uf)(x) \vee (vf)(x)\\
    &= (uf \vee vf)(x)
  \end{align*}
  and 
  \begin{align*}
    (u(f \vee g))(x) &= u((f \vee g)(x))\\
    &= u(f(x) \vee g(x))\\
    &= u(f(x)) \vee u(g(x))
  \end{align*}
  for all $u,v\in V$, $f,g \in \E$ and $x \in [0,\infty]_\bot$, since
  $u$ preserves the order.  Finally, $( \bbot\circ f)( x)= \FALSE=
  \bbot( x)$ and $( u\circ \bbot)( x)= \FALSE= \bbot( x)$ for all $u \in
  \V$, $f \in \E$ and $x \in [0,\infty]_\bot$.  This shows that $\V$ is
  a right $\E$-semimodule.  The other claims follow from the embedding
  of $\V$ into $\E$. \qed
\end{proofapp}

\begin{proofapp}[of Lemma~\ref{le:prod_is_prod}]
  To show the first equality, let $x_0\in[ 0, \infty]_\bot$ and set $x_{
    n+ 1}= f_n( x_n)$ for $n\in \Nat$.  If there is $n\in \Nat$ for
  which $x_n= \bot$, then $( \prod_{ i= 0}^\infty f_i)( x_0)= \FALSE$ by
  definition, and either $f_0( x_0)= \bot$ and consequently $( \prod_{
    i= 1}^\infty f_i)( f_0( x_0))= \FALSE$, or $f_0( x_0)\ne \bot$, but
  then also $( \prod_{ i= 1}^\infty f_i)( f_0( x_0))= \FALSE$.  If, on
  the other hand, $x_n\ne \bot$ for all $n\in \Nat$, then $( \prod_{ i=
    0}^\infty f_i)( x_0)= \TRUE$ and also $( \prod_{ i= 1}^\infty f_i)(
  f_0( x_0))= \TRUE$.  The second equality can be shown by similarly
  easy arguments. \qed
\end{proofapp}

\begin{proofapp}[of Proposition~\ref{pr:smodpair}]
  We need two technical lemmas in this proof.

  \begin{lemma}
    \label{le:smod-supseq}
    For all $u\in \V$, $n\in \Nat$, and $f_1,\dotsc, f_{ n+ 1},
    g_1,\dotsc, g_n\in \E$, we have $u f_{ n+ 1} g_n^* f_n\dotsm g_1^*
    f_1= \sup_{ k_1,\dotsc, k_n\in \Nat} u f_{ n+ 1} g_n^{ k_n}
    f_n\dotsm g_1^{ k_1} f_1$.
  \end{lemma}

  \begin{proof}%[of Lemma~\ref{le:smod-supseq}]
    By induction on $n$. For $n = 0$ our claim is clear. Suppose now
    that $n >0$ and that the claim holds for all $m < n$. Then
    \begin{align*}
      vf_{n+1}g_n^*\cdots g_1^*f_1 &= (\sup\{vf_{n+1}g_n^{k_n}\cdots
      g_2^{k_2}f_2\mid k_2,\ldots,k_n\in \Nat\})g_1^*f_1\\
      &= \sup\{vf_{n+1}g_n^{k_n}\cdots g_2^{k_2}f_2g_1^*\mid
      k_2,\ldots,k_n\in \Nat\}f_1\\
      &= \sup\big\{\sup\{ vf_{n+1}g_n^{k_n}\cdots
      g_2^{k_2}f_2g_1^{k_1}\mid k_1\geq
      0\}\bigmid k_2,\ldots,k_n \in \Nat\big\}f_1\\
      &= \sup\{vf_{n+1}g_n^{k_n}\cdots g_2^{k_2}f_2g_1^{k_1}\mid
      k_1,\ldots,k_n
      \in \Nat\}f_1\\
      &= \sup \{vf_{n+1}g_n^{k_n}\cdots g_1^{k_1}f_1\mid k_1,\ldots,k_n
      \in \Nat\}.
    \end{align*}
    Here the first equality is by induction hypothesis, the second by
    Lemma~\ref{le:distr}, the third by Lemma~\ref{le:distl}, and the
    fifth again by Lemma~\ref{le:distr}.  \qed
  \end{proof}

  \begin{lemma}
    \label{le:omega_sup}
    For all $f, g\in \E$, $( f\vee g)^\omega= \sup_{ h_0,
      h_1,\dotsc\in\{ f, g\}} \prod_{ i= 0}^\infty h_i$ and $( g
    f^*)^\omega= \sup_{ k_0, k_1,\dotsc\in \Nat} \prod_{ i= 0}^\infty( g
    f^{ k_i})$.
  \end{lemma}

  \begin{proof}%[of Lemma~\ref{le:omega_sup}]
    To show the first claim, note first that as infinite product
    preserves the order, we have $\sup_{h_i\in\{ f,g\}} \prod_{ i=
      0}^\infty h_i\le (f \vee g)^\omega$.  To complete the proof, we
    have to show that whenever $(f\vee g)^\omega (x) = \TRUE$ for some
    $x$, then $( \sup_{h_i\in\{ f,g\}} \prod_{ i= 0}^\infty h_i)(x) =
    \TRUE$.  So suppose that $(f \vee g)^\omega(x) = \TRUE$. Then we
    must have $(f \vee g)(x) \geq x$ and $x \neq \bot$, so that either
    $f(x) \geq x$ or $g(x) \geq x$.  Without loss of generality we can
    assume $f( x)\ge x$, but then $f^\omega(x) = \TRUE$ and $(
    \sup_{h_i\in\{ f,g\}} \prod_{ i= 0}^\infty h_i)(x)\ge f^\omega( x)=
    \TRUE$.

    For the second claim, it suffices to prove that if for some $x$,
    $(gf^*)^\omega(x) =\TRUE$, then there is a sequence $k_0,k_1,\ldots$
    such that $( \prod_{ i= 0}^\infty( g f^{ k_i}))(x) = \TRUE$.
    Suppose that $(gf^*)^\omega(x) =\TRUE$. Then $x \neq \bot$ and
    $gf^*(x) \geq x$.  We know that $gf^*(x) = \sup_n gf^n(x)$. Thus
    $\sup_n gf^n(x) \geq x$.  If $f(x) \leq x$, then the sequence
    $(gf^n(x))_n$ is decreasing, so that $\sup_n gf^n(x) = g(x) \geq
    x$. We thus have $g^\omega(x) = \TRUE$, and $( \prod_{ i= 0}^\infty(
    g f^{ k_i}))(x) = \TRUE$ holds when $k_i = 0$ for all $i$.  If, on
    the other hand, $f(x) > x$, then the sequence $(f^n(x))_n$ is
    strictly increasing with $\sup_n f^n(x) =\infty$. Since clearly $g
    \ne \bbot$, it follows that there is some $n$ with $gf^n(x) \geq x$.
    It follows now that $(gf^n)^\omega = ( \prod_{ i= 0}^\infty( g
    f^n))(x) \geq x$. \qed
  \end{proof}

  \bigskip\noindent Now for the proof of the proposition, we show first
  that $(gf)^\omega = (fg)^\omega f$ for all $f, g\in \E$.  Let $f, g\in
  \E$ and $x\in[ 0, \infty]_\bot$.  The result is clear for $x= \bot$,
  so let $x\ne \bot$.  Assume first that $( g f)( x)\ge x$, then $( g
  f)^\omega( x)= \TRUE$.  But by monotonicity, also $( f g)( f( x))= f(
  gf)( x)\ge f( x)$, hence $(( f g)^\omega f)( x)= \TRUE$.  Now assume
  that, instead, $( g f)( x)< x$, then $( g f)^\omega( x)= \FALSE$.
  By~\eqref{eq:deriv1}, also $( f g)( f( x))< f( x)$, so that $(( f
  g)^\omega f)( x)= \FALSE$.

  Next we show that $(f \vee g)^\omega = f^\omega(gf^*)^* \vee
  (gf^*)^\omega$.  We have
  \begin{align*}
    f^\omega(gf^*)^*&=
    \sup_n f^\omega(g f^*)^n \\
    &= \sup_n\sup_{k_n}\cdots \sup_{k_1} f^\omega gf^{k_n}\cdots gf^{k_1} \\
    &= \sup_{n,k_1,\ldots,k_n} ( \cdots ffgf^{k_n}\cdots gf^{k_1}) \\
    \intertext{and}
    (gf^*)^\omega &= 
    \sup_{k_1,k_2,\ldots} (\cdots  
    gf^{k_2}gf^{k_1})\,.
  \end{align*}
  Using these equalities, we conclude that 
  \begin{align*}
    f^\omega(gf^*)^* \vee (gf^*)^\omega
    &= \sup_{n,k_1,\ldots,k_n} (\cdots ffgf^{k_n}\cdots gf^{k_1}) \vee
    \sup_{k_1,k_2,\ldots} (\cdots gf^{k_2}gf^{k_1})\\
    &= \sup_{h_i\in\{ f,g\}}(\cdots h_2h_1)\\
    &= (f \vee g)^\omega\,.
  \end{align*} 

  \vspace*{-3ex}\qed
\end{proofapp}

\begin{proofapp}[of Lemma~\ref{le:onetransition}]
  We define
  \begin{equation*}
    T'=\big\{\big( s, \sup\{ f\mid( s, f, s')\in T\}, s'\big)\mid s,
    s'\in S\big\}\,;
  \end{equation*}
  note that $\sup \emptyset= \bbot$.  The idea of the construction is
  that if two transitions $( s, f, s'),( s, f', s')\in T$ are available,
  then it does not change the automaton's behavior if we replace them by
  a transition labeled with $f\vee f'$.

  To see this, we note that the semantic graph of $( S, T')$ is
  \emph{above} the one of $( S, T)$ in the following sense: If $( s,
  x)\to[ g]( s', x')$, with $x, x'\ne \bot$, is a transition in the
  semantic graph of $( S, T')$, then
  \begin{align*}
    x' &= ( \sup\{ f\mid( s, f, s')\in T\})( x) \\
    &= \sup\{ f( x)\mid( s, f, s')\in T\}\,,
  \end{align*}
  hence there is also a transition $\smash{( s, x)\to[ f]( s', x')}$ in
  the semantic graph of $( S, T)$.  On the other hand, if $( s, x)\to[
  f]( s', x')$ is a transition in the semantic graph of $( S, T)$, then
  we have a transition $s\to[ g] s'$ in $( S, T')$ with $g\ge f$, hence
  a transition $( s, x)\to[ g]( s', x'')$ in the semantic graph of $( S,
  T')$ with $x''\ge x'$.

  It is now clear that any accepting run (in the reachability or in the
  B{\"u}chi sense) in $( S, T')$ is also available in $( S, T)$;
  likewise, any accepting run in $( S, T)$ has one which is above it in
  $( S, T')$.  \qed
\end{proofapp}

\begin{proofapp}[of Theorem~\ref{th:buchi}]
  The statement in the theorem is equivalent to asserting that there is
  an accepting state $j$ which is reachable from $(s_0, x_0)$ with
  output energy $x= \transpose I^j T^* I^{ s_0}( x_0)$, and at which
  there is a non-trivial loop which does not lose energy,
  \ie~$\transpose I^j T T^* I^j( x)\ge x$.

  We shall need a few technical lemmas in the proof.  For a global state
  $q=( s, x)$, we write $\state( q)= s$ and $\val( q)= x$ below.

  \begin{lemma}
    \label{lemma_remove_cycle}
    Let $q_0\to[ f_1] q_1 \to[ f_2]\cdots\to[ f_n] q_n$ be a run of
    $(S,T)$ such that there are $i<j$ for which $q_i\leadsto q_j$ is a
    loop with $\val( q_i)\ge \val( q_j)$.  Then there are global states
    $q_{ j+ 1}',\dots, q_n'$ with $\state( q_k')= \state( q_k)$ for all
    $k$, and such that
    \begin{equation*}
      q_0\to[ f_1] q_1 \to[ f_2]\cdots\to[ f_i] q_i\to[
      f_{j+1}] q'_{j+1}\to[ f_{j+2}]\cdots\to[ f_n] q'_n
    \end{equation*}
    is a run of $(S,T)$ with $\val(q_n')\ge \val(q_n)$.
  \end{lemma}

  \begin{proof}%[of Lemma~\ref{lemma_remove_cycle}]
    By $\val( q_i)\ge \val( q_j)$ and~\eqref{eq:deriv1}, also $\val( q_{
      j+ 1}')\ge \val( q_{ j+ 1})$.  Inductive application
    of~\eqref{eq:deriv1} finishes the proof. \qed
  \end{proof}

  \begin{lemma}
    \label{lemma_iterate_cycle}
    Let $q_0\to[ f_1] q_1 \to[ f_2]\cdots\to[ f_n] q_n$ be a run of
    $(S,T)$ such that there are $i<j$ for which $q_i\leadsto q_j$ is a
    loop with $\val( q_j)> \val( q_i)$.  Then there are $\Delta> 0$ and,
    for each $k\geq 1$, global states $q_{ i+ 1}^l,\dots, q_j^l$ for all
    $1\le l\le k$, and $q^k_{j+1},\dots,q^k_n$, with $\state( q_m^l)=
    \state( q_m)$ for each $k$, $l$, such that
    \begin{equation*}
      q_0\to[ f_1]\cdots\to[ f_i] q_i \Big( \to[ f_{i+1}] q_{i+1}^l\to[
      f_{i+2}] \cdots\to[ f_j] q^l_j \Big)_{1\leq l\leq k}
      \to[ f_{j+1}] q^k_{j+1}\to[ f_{j+2}]\cdots\to[ f_n] q^k_n
    \end{equation*}
    is a run of $(S,T)$, and $\val(q^{ k+ 1}_n) \ge \val(q^k_n)+
    \Delta$.
  \end{lemma}

  \begin{proof}%[of Lemma~\ref{lemma_iterate_cycle}]
    Let $\Delta= \val( q_j)- \val( q_i)> 0$.  By~\eqref{eq:deriv1},
    $\val( q_j^{ l+ 1})- \val( q_j^l)\ge \Delta$ for all $l$, hence also
    $\val(q^{ k+ 1}_n)- \val(q^k_n)\ge \Delta$. \qed
  \end{proof}

  We say that a loop $q_i\leadsto q_j$ with $\val( q_j)> \val( q_i)$ as
  in Lemma~\ref{lemma_iterate_cycle} is \emph{energy producing};
  similarly, a loop with $\val( q_j)< \val( q_i)$ will be called
  \emph{energy consuming}.  As a corollary of the last lemma, the energy
  value at $\state( q_n)$ can be pushed arbitrarily high: for any $M\in
  \Realnn$, there exists $k$ such that the $k$-iteration $q_0\leadsto
  q_i(\to q_{ i+ 1}^l\to\cdots\to q_j^l)_{ 1\le 1\le k}\to q_{ j+
    1}^k\leadsto q_n^k$ has $\val( q_n^k)\ge M$.  We also remark that
  similar results are available for infinite runs; also with these we
  can remove loops which are not energy producing and iterate loops
  which are.

  \begin{lemma}
    \label{le:inf-decreasing}
    Let $f\in \E$ and $x\in \Realnn$. If $f( x)< x$, then $\lim_{ n\to
      \infty} f^n( x)= \bot$.
  \end{lemma}

  \begin{proof}%[of Lemma~\ref{le:inf-decreasing}]
    We have $x- f( x)= M> 0$.  Using~\eqref{eq:deriv1}, we see that $f^{
      n+ 1}( x)\le f^n( x)- M$ for all $n\in \Nat$.  Hence $( f^n( x))_{
      n\in \Nat}$ decreases without bound, so that there must be $N\in
    \Nat$ such that $f^n( x)= \bot$ for all $n\ge N$.  \qed
  \end{proof}

  \bigskip\noindent Now for the proof of the theorem, the backwards
  direction is clear, as the cycle at $j$ contains an accepting state
  and can be iterated indefinitely because of $\transpose I^j T T^* I^j(
  x)\ge x$.  For the forward direction, let $\rho$ be an infinite run
  from $( s_0, x_0)$ which visits $F$ infinitely often, then $\rho$ must
  contain an accepting state $s\le k$ infinitely often, so we can write
  \begin{equation*}
    \xymatrix@R=1ex{%
      \rho_{ \vphantom0} : s_0 \ar@{~>}[r]_(.6){ \rho_0} & s
      \ar@{~>}[r]_{ \rho_1} & 
      s \ar@{~>}[r]_{ \rho_2} & s \cdots
    }\,.
  \end{equation*}
  Now inductively for each $i\ge 1$, we modify $\rho$ as follows:
  \begin{itemize}
  \item If $\rho_i$ is not energy consuming, we can use
    Lemma~\ref{lemma_iterate_cycle} to replace $\rho$ by the run
    $\rho_0\cdots \rho_i \rho_i^1 \rho_i^2\cdots$, where each $\rho_i^j$
    is a loop over the same cycle as $\rho_i$.  We have constructed a
    run which consists of a finite prefix and a loop.
  \item In case $\rho_i= \mu_1 \mu_2 \mu_3$ strictly contains a loop
    $\mu_2$ which is energy producing, we can use
    Lemma~\ref{lemma_iterate_cycle} to iterate $\mu_2$ sufficiently
    often so that the combined run $\tilde \rho_i= \mu_1 \mu_2
    \mu_2^1\cdots \mu_2^N \mu_3'$, where each $\mu_2^j$ is a loop over
    the same cycle as $\mu_2$ and $\mu_3'$ over the same path as
    $\mu_3$, is not energy consuming.  Using
    Lemma~\ref{lemma_iterate_cycle} again, we can now replace $\rho$ by
    the run $\rho_0\cdots \rho_{ i- 1} \tilde \rho_i \tilde \rho_i^1
    \tilde \rho_i^2\cdots$, where each $\tilde \rho_i^j$ is a loop over
    the same cycle as $\tilde \rho_i$.  This run consists of a finite
    prefix and a loop.
  \item Any loops properly contained in $\rho_i$ which are not energy
    producing can be removed by Lemma~\ref{lemma_remove_cycle}.
  \end{itemize}

  Assume, for the sake of contradiction, that the above induction never
  finishes.  Then we have constructed a run $\rho'= \rho_0 \rho_1'
  \rho_2'\cdots$ in which each $\rho_i'$ is an energy consuming simple
  loop from $s$ to $s$.  There are only finitely many simple cycles from
  $s$ to $s$, thus one of them appears infinitely often as a cycle
  underlying a simple loop in $\rho'$.  Call this cycle $\pi$ and let
  $\sigma_1, \sigma_2,\dots$ be the loops in $\rho'$ over it.

  We can hence write $\rho'= \rho_0 \mu_1 \sigma_1 \mu_2
  \sigma_2\cdots$.  All loops $\mu_i$ are energy consuming, hence can be
  removed from $\rho'$ by Lemma~\ref{lemma_remove_cycle}.  We have
  constructed an infinite run $\rho_0 \sigma_1' \sigma_2'\cdots$ in $(
  S, T)$ in which each $\sigma_i'$ is energy consuming.  The sequence of
  energy values after each iteration $\sigma_i'$ is given by $( f_\pi^n(
  f_{ \rho_0}( x_0)))_{ n\in \Nat}$, but by
  Lemma~\ref{le:inf-decreasing} and as $f_\pi( f_{ \rho_0}( x_0))< f_{
    \rho_0}( x_0)$ ($f_\pi( f_{ \rho_0}( x_0))$ is the energy value
  after the first loop $\sigma_1'$), $\lim_{ n\to \infty} f_\pi^n( f_{
    \rho_0}( x_0))= f_\pi^\omega( f_{ \rho_0}( x_0))= \bot$, a
  contradiction. \qed
\end{proofapp}

\begin{figure}[tbp]
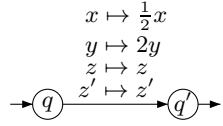

  \centering
  \begin{gpicture}(28,16)(0,-16)
    \node[Nmarks=i,ilength=3,Nw=4,Nh=4](q)(5,-14){$q$}
    \node[Nmarks=f,flength=3,Nw=4,Nh=4](q2)(23,-14){$q'$}
    
    \drawedge(q,q2){}
    \put(10,-3){\footnotesize{$x\mapsto \frac{1}{2} x$}}
    \put(10,-7){\footnotesize{$y\mapsto 2 y$}}
    \put(10,-10){\footnotesize{$z\mapsto z$}}
    \put(9,-13){\footnotesize{$z'\mapsto z'$}}
  \end{gpicture}
  \caption{Module for $(q,c_1\texttt{++},q')$}
  \label{figure_module_increment_c1}
\end{figure}

\begin{figure}[tbp]
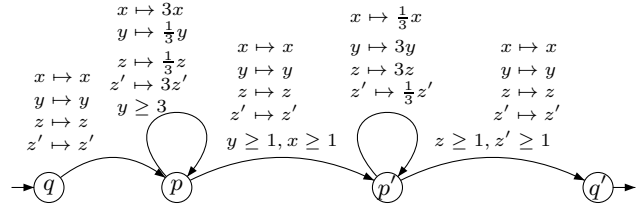

  \centering
  \begin{gpicture}(96,27)(0,-27)
    \node[Nmarks=i,ilength=3,Nw=4,Nh=4](q)(5,-25){$q$}
    \node[Nw=4,Nh=4](p)(22,-25){$p$}
    \node[Nw=4,Nh=4](p2)(50,-25){$p'$}
    \node[Nmarks=f,flength=3,Nw=4,Nh=4](q2)(78,-25){$q'$}
    
    \drawedge[curvedepth=4](q,p){}
    \drawloop(p){}
    \drawedge[curvedepth=4](p,p2){\scriptsize{$y\geq 1, x\geq 1$}}
    \drawloop(p2){}
    \drawedge[curvedepth=4](p2,q2){\scriptsize{$z\geq 1,z'\geq 1$}}
    
    \put(3,-11){\scriptsize{$x\mapsto x$}}
    \put(3,-14){\scriptsize{$y\mapsto y$}}
    \put(3,-17){\scriptsize{$z\mapsto z$}}
    \put(2,-20){\scriptsize{$z'\mapsto z'$}}
    
    \put(14,-2){\scriptsize{$x\mapsto 3 x$}}
    \put(14,-5){\scriptsize{$y\mapsto \frac{1}{3}y$}}
    \put(14,-9){\scriptsize{$z\mapsto \frac{1}{3} z$}}
    \put(13,-12){\scriptsize{$z'\mapsto 3z'$}}
    \put(14,-15){\scriptsize{$y\geq 3$}}
    
    \put(30,-7){\scriptsize{$x\mapsto x$}}
    \put(30,-10){\scriptsize{$y\mapsto y$}}
    \put(30,-13){\scriptsize{$z\mapsto z$}}
    \put(29,-16){\scriptsize{$z'\mapsto z'$}}
    
    \put(45,-3){\scriptsize{$x\mapsto \frac{1}{3} x$}}
    \put(45,-7){\scriptsize{$y\mapsto 3y$}}
    \put(45,-10){\scriptsize{$z\mapsto 3 z$}}
    \put(45,-13){\scriptsize{$z'\mapsto \frac{1}{3} z'$}}
    
    \put(65,-7){\scriptsize{$x\mapsto x$}}
    \put(65,-10){\scriptsize{$y\mapsto y$}}
    \put(65,-13){\scriptsize{$z\mapsto z$}}
    \put(64,-16){\scriptsize{$z'\mapsto z'$}}
  \end{gpicture}
  \caption{Module for $(q,c_1\texttt{=0?},q')$}
  \label{figure_module_zero_test_c1}
\end{figure}

\begin{proofapp}[of Theorem~\ref{th:undec4dim}]
  We reduce from the halting problem for 2-counter machines.  We use two
  energy variables $x,y$ to encode the values of the counters $c_1$ and
  $c_2$, and we use two additional energy variables $z,z'$ for storing
  temporary information needed for encoding zero tests and
  decrementation operations of the 2-counter machine.  The initial value
  of all energy variables is $1$.  Let $(q, c_1\texttt{++}, q')$ be a
  transition of the 2-counter machine that increments the value of the
  first counter.  The module simulating this transition is shown in
  Fig.~\ref{figure_module_increment_c1}.  The intuition is that $x$ and
  $y$ encode the counter values as $x=1/(2^{c_1}3^{c_2})$ and
  $y=2^{c_1}3^{c_2}$, respectively.  Hence $x$ is divided and $y$ is
  multiplied by $2$ to encode the incrementation of $c_1$.  Likewise, an
  incrementation of $c_2$ is encoded by dividing $x$ and multiplying $y$
  by $3$.  Let $(q,c_1\texttt{=0}?, q')$ be a zero test transition for
  $c_1$.  The corresponding module is shown in
  Fig.~\ref{figure_module_zero_test_c1}.  Note that in order to take the
  transition from $p$ to $p'$, the values of both $x$ and $y$ have to be
  greater than or equal to $1$.  This is the case if and only if we loop
  in $p$ exactly $c_2$ times and $c_1=0$.  In $p'$ we have to loop for
  the same number of times as in $p$ to restore the original values of
  $x$ and $y$. For this we use the information stored in $z$ and $z'$
  together with the lower bound restrictions at the transition from $p'$
  to $q'$.  For the simulation of zero test transitions for $c_2$, we
  replace $\frac{1}{3}$ and $3$ by $\frac{1}{2}$ and $2$, respectively.
  The idea for encoding a decrementation operation of $c_1$ and $c_2$ is
  similar.  \qed
\end{proofapp}

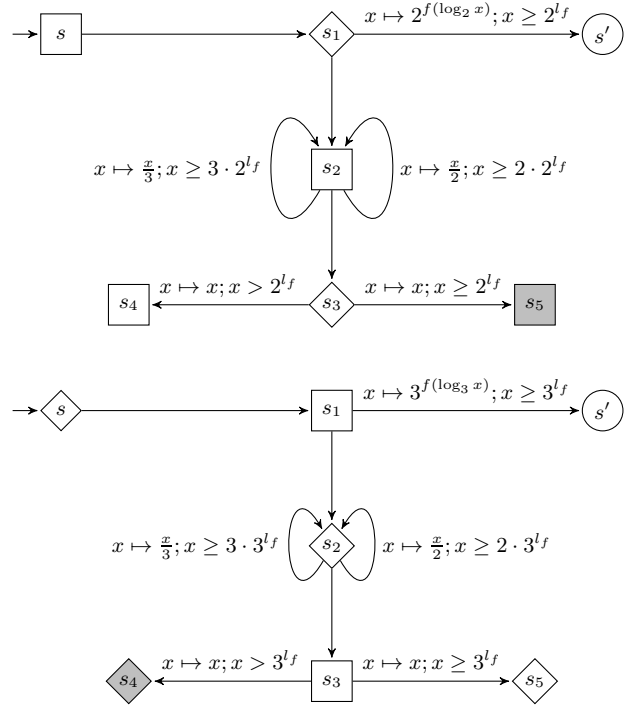
\begin{figure}[t]
  \centering
  \begin{tikzpicture}[->,>=stealth',shorten >=1pt,auto,node
    distance=2.0cm,initial text=,scale=0.9,transform shape]
    \tikzstyle{every node}=[font=\small]
    \tikzstyle{every state}=[fill=white,shape=circle,inner
    sep=.5mm,minimum size=6mm,pB/.style=diamond,pA/.style=rectangle]
    \node[state,initial,pA] (s) at (-2,0) {$s$};
    \node[state,pB] (s1) at (2,0) {$s_1$};
    \node[state] (s') at (6,0) {$s'$};
    \node[state,pA] (s2) at (2,-2) {$s_2$};
    \node[state,pB] (s3) at (2,-4) {$s_3$};
    \node[state,pA] (s4) at (-1,-4) {$s_4$};
    \node[state,pA,fill=lightgray] (s5) at (5,-4) {$s_5$};
    \path (s) edge (s1);
    \path (s1) edge node{$x\mapsto 2^{ f( \log_2 x)}; x\ge 2^{
        l_f}$} (s');
    \path (s2) edge[out=240,in=120,loop] node{$x\mapsto \frac x 3;
      x\ge 3\cdot 2^{ l_f}$} (s2);
    \path (s2) edge[out=300,in=60,loop] node[right]{$x\mapsto \frac x
      2; x\ge 2\cdot 2^{ l_f}$} (s2);
    \path (s2) edge (s3);
    \path (s3) edge node{$x\mapsto x; x\ge 2^{ l_f}$} (s5);
    \path (s1) edge (s2);
    \path (s3) edge node[above]{$x\mapsto x; x> 2^{
        l_f}$} (s4);
  \end{tikzpicture}
  
  \vspace*{4ex}
  
  \begin{tikzpicture}[->,>=stealth',shorten >=1pt,auto,node
    distance=2.0cm,initial text=,scale=.9,transform shape]
    \tikzstyle{every node}=[font=\small]
    \tikzstyle{every state}=[fill=white,shape=circle,inner
    sep=.5mm,minimum size=6mm,pB/.style=diamond,pA/.style=rectangle]
    \node[state,initial,pB] (s) at (-2,0) {$s$};
    \node[state,pA] (s1) at (2,0) {$s_1$};
    \node[state] (s') at (6,0) {$s'$};
    \node[state,pB] (s2) at (2,-2) {$s_2$};
    \node[state,pA] (s3) at (2,-4) {$s_3$};
    \node[state,pB,fill=lightgray] (s4) at (-1,-4) {$s_4$};
    \node[state,pB] (s5) at (5,-4) {$s_5$};
    \path (s) edge (s1);
    \path (s1) edge node{$x\mapsto 3^{ f( \log_3 x)}; x\ge 3^{
        l_f}$} (s');
    \path (s2) edge[out=240,in=120,loop] node{$x\mapsto \frac x 3;
      x\ge 3\cdot 3^{ l_f}$} (s2);
    \path (s2) edge[out=300,in=60,loop] node[right]{$x\mapsto \frac x
      2; x\ge 2\cdot 3^{ l_f}$} (s2);
    \path (s2) edge (s3);
    \path (s3) edge node{$x\mapsto x; x\ge 3^{ l_f}$} (s5);
    \path (s1) edge (s2);
    \path (s3) edge node[above]{$x\mapsto x; x> 3^{
        l_f}$} (s4);
  \end{tikzpicture}
  \caption{Conversion of two types of edges in $( S, T)$.  Top: an edge
    $( s,( f, \id), s')$ from a player-$A$ state $s$; bottom: an edge $(
    s,( \id, f), s')$ from a player-$B$ state $s$.  Player-$A$ states
    are depicted using squares, player-$B$ states are diamonds.
    Accepting states have a gray background color.  The ownership of
    state $s'$ is unchanged.}
  \label{fi:convert2to1-A1}
\end{figure}

\begin{proofapp}[of Theorem~\ref{th:flatgame}]
  We show a reduction from reachability games on $2$-dimensional
  $\Eint$-automata to reachability games on $1$-dimensional
  $\Eflatpw$-automata.  Let $( S, T)$ be a $2$-dimensional
  $\Eint$-automaton.  By inserting extra states (and transitions) if
  necessary, we can assume that for any $( s,( f, g), s')\in T$, either
  $f= \id$ with $l_f= 0$, or $g= \id$ with $l_g= 0$.  We build an energy
  automaton $( S', T')$.

  Let $( s,( f, \id), s')\in T$ be a player-$A$ transition (\ie~$s\in
  S_A$) in $( S, T)$ (with lower bound $l_f$ as usual), then we model
  this in $( S', T')$ using $s$, $s'$ and the following new states and
  transitions; see Figure~\ref{fi:convert2to1-A1} for a pictorial
  description.
  \begin{itemize}
  \item player-$A$ states: $s_2$, $s_4$, $s_5$ (accepting); player-$B$
    states: $s_1$, $s_3$
  \item transitions:
    \begin{itemize}
    \item $( s,[ x\mapsto x; x\ge 0], s_1)$; $( s_1,[ x\mapsto 2^{ f(
        \log_2 x)}; x\ge 2^{ l_f}], s')$
    \item $( s_1,[ x\mapsto x; x\ge 0], s_2)$
    \item $( s_2,[ x\mapsto \frac x 3; x\ge 3\cdot 2^{ l_f}], s_2)$; $(
      s_2,[ x\mapsto \frac x 2; x\ge 2\cdot 2^{ l_f}], s_2)$
    \item $( s_2,[ x\mapsto x; x\ge 0], s_3)$
    \item $( s_3,[ x\mapsto x; x> 2^{ l_f}], s_4)$; $( s_3,[ x\mapsto x;
      x\ge 2^{ l_f}], s_5)$
    \end{itemize}
  \end{itemize}
  Note that $s_4$ is a deadlock state, hence player~$A$ loses the
  reachability game if $s_4$ is reached.  Similarly, she wins if $s_5$
  is reached.

  The intuition is that the new energy variable $x$ encodes the two old
  ones as $x= 2^{ x_1} 3^{ x_2}$.  If player~$A$ wants to bring $( S',
  T')$ from $s$ to $s'$, and commits to this by taking the transition
  $s\to s_1$, she may be interrupted by player~$B$ taking the $s_1\to
  s_2$ transition.  Here player~$A$ has to prove that $x_1$ was really
  $\ge l_f$, by using the loops at $s_2$ to bring $x$ to the precise
  value $2^{ l_f}$.  If she manages this, then player~$B$ has only the
  $s_3\to s_5$ transition available in $s_3$, hence player~$A$ wins.
  Otherwise, player~$B$ wins.

  The conversions of other types of transitions are similar.  One can
  easily see that player~$A$ can reach a state in $F$ in the original
  energy automaton $( S, T)$ if, and only if, she can reach a state in
  $F$, or one of the new accepting states, in the new automaton $( S',
  T')$.

  We miss to argue that all energy functions in $( S', T')$ are
  piecewise affine.  Looking at the defined modules, we see that this is
  the case except perhaps for the functions defined as $g_2( x)= 2^{ f(
    \log_2 x)}$ and $g_3( x)= 3^{ f( \log_3 x)}$.  However, $f$ is an
  integer update function, so that $f( x)= x+ k$ for some $k\in \Int$;
  hence $g_2( x)= 2^k x$ and $g_3( x)= 3^k x$, which are indeed
  piecewise affine.  \qed
\end{proofapp}

\end{document}